%% file: main.tex
\documentclass[letterpaper,twocolumn,10pt]{article}
\usepackage{style-files/usenix2019_v3}

\usepackage[numbers,sort&compress,square]{natbib}
\usepackage[normalem]{ulem} % Defines strikeout Breaks bibtex

\usepackage{amsmath,amsfonts,amsthm}
\usepackage{algorithm}
\usepackage[noend]{algpseudocode}
\usepackage[shortlabels]{enumitem}
\usepackage{graphicx}
\usepackage{textcomp}
\usepackage{xcolor}
\usepackage{url}
\usepackage{xspace}
\usepackage{cleveref}
\usepackage{pifont}
\usepackage{booktabs}
\usepackage{wrapfig}
\usepackage[font=footnotesize]{caption}
\usepackage{thmtools}
\usepackage{thm-restate}
\usepackage{tabularx}
\usepackage{enumitem}
\usepackage{caption}
\usepackage{subcaption}
\usepackage{soul}
\usepackage{nopageno}
\usepackage{mathtools}

\usepackage{setspace}

  \setlist[itemize]{leftmargin=*}
  \setlist[enumerate, 1]{1.} %{leftmargin=*}
  \setlist{itemsep=0pt,parsep=2pt}             % more compact lists
  \usepackage[subtle,wordspacing=normal]{savetrees} % Saves space
\usepackage{todonotes}

\def\cameraReady{} % set to true

\begin{document}

\input{macros.tex}
\graphicspath{{figures/}}

\date{}

%\title{\sysname: High Throughput BFT\\ Meets Low Latency on A DAG}
\title{\sysname: High Throughput DAG BFT Can Be Fast and Robust!}
%\title{\sysname: High Throughput Certified DAG BFT Can Be Fast!}
%\title{\sysname: Certified DAG BFT Can Be Quick and Fast!}

\author{
{\rm Balaji Arun}\\
Aptos Labs
\and
{\rm Zekun Li}\\
Aptos Labs
\and
{\rm Florian Suri-Payer\textsuperscript{*}}\\  % Footnote added here
Cornell University
\and
{\rm Sourav Das\textsuperscript{*}}\\  % Refers to the footnote above
UIUC
\and
{\rm Alexander Spiegelman}\\
Aptos Labs
}

%\author{ID number: 1241}

\maketitle

\begingroup
\renewcommand{\thefootnote}{*}
\footnotetext{The work was done while the authors were interns at Aptos Labs.}
\endgroup

  %  \sasha{set submission ID}
\input{sections/0-abstract.tex}

\input{sections/1-introduction-new}

\input{sections/1.5-model}
\input{sections/2-background-new}
\input{sections/3-protocol}

\input{sections/4-proofs}
\input{sections/5-practical-new}

%\input{sections/5-practical}
%\input{sections/6-implementation}
\input{sections/7-evaluation}
\input{sections/8-related}

\input{sections/9-conclusion}

\ifdefined\cameraReady
%\section*{Acknowledgments}

%We thank the shepherd and the reviewers for their feedback on the paper. 

\fi

\bibliographystyle{ACM-Reference-Format}
\bibliography{references}

\end{document}

%% file: macros.tex
% ====== comments ======

\ifdefined\cameraReady
\fi

\newcommand{\textcomment}[2]{\textbf{#1:} #2}
\newcommand\rati[1]{\todo[color=yellow,inline]{\textcomment{Rati}{#1}}}
\newcommand\sasha[1]{\todo[color=red,inline]{\textcomment{Sasha}{#1}}}
\newcommand\zekun[1]{\todo[color=brown,inline]{\textcomment{Zekun}{#1}}}
\newcommand\balaji[1]{\todo[color=purple,inline]{\textcomment{Balaji}{#1}}}
\newcommand\fs[1]{\todo[color=cyan,inline]{\textcomment{Florian}{#1}}}
\newcommand{\changebars}[2]{{[check]}[{\em \color{red}#1}] [{\sout {\em #2}]}}
\newcommand{\sd}[1]{\textcolor{blue}{{\bf Sourav:} #1}}

\newcommand\com[1]{}

% ====== systems ======
\newcommand\sysname{Shoal++\xspace}
\newcommand\protocolname{Mako\xspace}

\newcommand\chainspace{Chainspace\xspace}
\newcommand\ethereum{Ethereum\xspace}
\newcommand\hyperledger{Hyperledger\xspace}
\newcommand\omniledger{Omniledger\xspace}
\newcommand\rapidchain{RapidChain\xspace}
\newcommand\coconut{Coconut\xspace}
\newcommand\bft{BFT\xspace}
\newcommand\rscoin{RSCoin\xspace}
\newcommand\bftsmart{\textsc{bft-SMaRt}\xspace}
\newcommand\byzcuit{Byzcuit\xspace}
\newcommand\bitcoin{Bitcoin\xspace}
\newcommand\bitcoinng{Bitcoin-NG\xspace}
\newcommand\cosi{CoSi\xspace}
\newcommand\byzcoin{ByzCoin\xspace}
\newcommand\elastico{Elastico\xspace}
\newcommand\algorand{Algorand\xspace}
\newcommand\hyperledgerfabric{Hyperledger Fabric\xspace}
\newcommand\pbft{PBFT\xspace}
\newcommand\bftlong{Byzantine Fault-Tolerant\xspace} 
\newcommand\solidus{Solidus\xspace}
\newcommand\hashgraph{Hashgraph\xspace}
\newcommand\avalanche{Avalanche\xspace}
\newcommand\blockmania{Blockmania\xspace}
\newcommand\stellar{Stellar\xspace}
\newcommand\libra{Libra\xspace}
\newcommand\librabft{LibraBFT\xspace}
\newcommand\move{Move\xspace}
\newcommand\hotstuff{HotStuff\xspace}
\newcommand\vanillahs{Vanilla-HotStuff\xspace}
\newcommand\batchedhs{Batched-HotStuff\xspace}

%  ===== custom notations ======
\newcommand{\keyword}[1]{\normalfont \texttt{#1}}
\newcommand{\accounts}{\keyword{accounts}}

\newcommand{\transfer}{O}
\newcommand{\cert}{C}
\newcommand{\sync}{S}
\newcommand{\account}{a}
\newcommand{\authority}{\alpha}

%  ===== formatting ======
\renewcommand{\paragraph}[1]{\vspace{2mm}\noindent\textbf{#1}.\xspace}

% Abbreviations
\newcommand{\cf}{cf.\@\xspace}
\newcommand{\vs}{vs.\@\xspace}
\newcommand{\etc}{etc.\@\xspace}
\newcommand{\ala}{ala\@\xspace}
\newcommand{\wrt}{w.r.t.\@\xspace}
\newcommand{\etal}{\textit{et al.}\@\xspace}
\newcommand{\eg}{\textit{e.g.}\@\xspace}
\newcommand{\ie}{\textit{i.e.}\@\xspace}
\newcommand{\via}{\textit{via}\@\xspace}
\newcommand{\defacto}{\textit{de facto}\@\xspace}

% Theorems
\newtheorem{assumption}{Security Assumption}
\newtheorem{property}{Property}
\newtheorem{theorem}{Theorem}
\newtheorem{lemma}{Lemma}

% For inline section titles
\newcommand\inlinesection[1]{{\bf #1.}}

\def\first{({i})\xspace}
\def\second{({ii})\xspace}
\def\third{({iii})\xspace}
\def\fourth{({iv})\xspace}
\def\fifth{({v})\xspace}
\def\sixth{({vi})\xspace}

\newcommand{\one}{({i})\xspace}
\newcommand{\two}{({ii})\xspace}
\newcommand{\three}{({iii})\xspace}
\newcommand{\four}{({iv})\xspace}
\newcommand{\five}{({v})\xspace}
\newcommand{\six}{({vi})\xspace}

% Colors
\definecolor{verylightgray}{gray}{0.9}

% Markers
\newcommand\vgap{\vskip 2ex}
\newcommand\marker{\vgap\ding{118}\xspace}
\def\na{--}
\def\unsure{?}
\def\missing{$!$}
\newcommand{\yes}{\ding{51}}
\newcommand{\no}{\ding{55}}
\DeclareRobustCommand\pie[1]{
\tikz[every node/.style={inner sep=0,outer sep=0, scale=1.5}]{
\node[minimum size=1.5ex] at (0,-1.5ex) {}; 
\draw[fill=white] (0,-1.5ex) circle (0.75ex); \draw[fill=black] (0.75ex,-1.5ex) arc (0:#1:0.75ex); 
}}
\def\L{\pie{0}} % Low
\def\M{\pie{-180}} % Medium
\def\H{\pie{360}} % High

% Listings

%Algorithmicx
\algdef{SE}[Upon]{Upon}{EndUpon}[1]{\textbf{upon}\ #1\ \algorithmicdo}{\algorithmicend\ \textbf{}}%
\algtext*{EndUpon}

\algdef{SE}[Receiving]{Receiving}{EndReceiving}[1]{\textbf{upon
receiving}\ #1\ \algorithmicdo}{\algorithmicend\ \textbf{}}%
\algtext*{EndReceiving}
\newcommand\StateX{\Statex\hspace{\algorithmicindent}}
\algrenewcommand\textproc{}% Used to be \textsc

%% file: sections/0-abstract.tex
\begin{abstract}
Today's practical partially synchronous Byzantine Fault Tolerant consensus protocols trade off low latency and high throughput. On the one end, traditional BFT protocols such as PBFT and its derivatives optimize for latency. 
They require, in fault-free executions, only 3 message delays
to commit, the optimum for BFT consensus. However, this class of protocols typically relies on a single leader, hampering throughput scalability.
On the other end, a new class of so-called DAG-BFT protocols demonstrates how to achieve highly scalable throughput by separating data dissemination from consensus, and using every replica as proposer. Unfortunately, existing DAG-BFT protocols pay a steep latency premium, requiring on average 10.5 message delays to commit transactions.

This work aims to soften this tension, and proposes \sysname, a novel DAG-based BFT consensus system that offers the throughput of DAGs while reducing end-to-end consensus commit latency to an average of 4.5 message delays.
Our empirical findings are encouraging, showing that \sysname achieves throughput comparable to state-of-the-art DAG BFT solutions while reducing latency by up to 60\%, even under less favorable network and failure conditions.

%\fs{Give some numbers.}
%\balaji{remove fast and quick?}
%\sasha{yes, done}

\end{abstract}

%% file: sections/1-introduction-new.tex
\vspace{-1.75mm}
\section{Introduction}
This paper presents \sysname, a novel partially synchronous DAG-based Byzantine Fault Tolerant (BFT) consensus protocol that matches the throughput of contemporary DAG-BFT protocols while significantly reducing latency, narrowing the gap with the theoretical optimum.%, close to the theoretical optimum.

%Pioneered by Lamport~\cite{lamport2019byzantine}, 
BFT consensus offers the attractive abstraction of a single, always available server, that remains correct even as a subset of replicas may fail arbitrarily. Fueled by recent advances to simplicity and practicality~\cite{castro2002practical, hotstuff, narwhaltusk, bullshark}, interest and adoption of BFT protocols is skyrocketing: it is at the core of aspiring multi-national projects such as the digital Euro~\cite{digital-euro}, confidential data sharing frameworks~\cite{ccf}
and the growing Web3 industry. Already today, numerous blockchain companies~\cite{aptos, algorand, ava, sui} are actively deploying BFT protocols~\cite{bullshark, spiegelman2023shoal, chen2019algorand, snowflake, jolteon}, supporting millions of daily worldwide users.  

Following the seminal PBFT~\cite{castro2002practical} protocol, most traditional BFT consensus protocols~\cite{hotstuff, zyzzyva, sbft} operate in a primary-backup like regime in which a single, dedicated leader replica proposes an ordering of commands, and follower replicas vote to ensure consistency and durability. 

As demonstrated by PBFT, this approach can enjoy excellent latency in fault-free executions, requiring only three message delays (md) to order a leader's proposal, the optimum achievable for BFT protocols.
%\fs{pbft fairness discussion comments commented out}
% \sasha{I want to be careful here because, in an apples-to-apples comparison, PBFT requires 4 message delays - one to forward txns to the leader. Not sure if here is the place or later to mention this.}\fs{I think it helps our cause in the intro to make pbft sound as good as possible, to highlight how big the gap is that we try to bridge; LATER after presenting our system we can discuss are more fine grained apples to apple comparisons}
A dedicated leader, however, poses not only a single point of failure, but constrains the achievable throughput to the processing and networking bandwidth available to a single server.
To improve scalability, and meet the high throughput demands of modern applications, recent proposals emphasize maximizing resource utilization and balancing load across \textit{all} replicas. 

Initial approaches propose multi-log frameworks~\cite{MirBFT} that partition the request space, and operate multiple consensus instances in parallel (each led by a different replica) whose logs are carefully intertwined. More recent designs take scalability a notch further by entirely separating data dissemination and consensus. 
Pioneered by Aleph~\cite{aleph} and DAG-Rider~\cite{allyouneed}, and made practical by Narwahl~\cite{narwhaltusk} a new class of so-called DAG-BFT protocols have risen to popularity. 

DAG-BFT protocols employ a structured, highly parallel, asynchronously responsive data dissemination layer as a backbone, and layer consensus atop. Data dissemination proceeds through a series of \textit{rounds} in which each replica proposes new transaction batches via reliable broadcast (RB); each proposal (DAG \textit{node}) references a subset of preceding proposals (DAG \textit{edges}), resulting in a Directed Acyclic Graph (DAG) of temporally (or \textit{causally}) ordered proposals. The consensus layer reaches agreement on prefixes of the DAG, and derives an ordering for all causally referenced proposals.
% by locally interpreting the DAG structure.  (no communication is required). 

State-of-the-art DAG protocols~\cite{bullshark, spiegelman2023shoal} avoid running an explicit consensus protocol, and instead embed the consensus process onto the DAG directly: they ensure that all replicas converge on the same view of the DAG, and interpret it's structure locally to implicitly determine commitment. 
%\rati{comment re:FLP and consensus simulate..}
% Is it consensus that simulates? also consensus simulates for live consensus.
% Also nit, but FLP doesn't say leaders are required? it just has conditions
% under which consensus can't be solved.
% The tie-breaking required for live consensus is accomplished by assigning..?
Consensus simulates leaders by assigning pre-determined "fixpoint" nodes in the DAG (henceforth \textit{anchors}) and confirms agreement of an anchor by waiting until it has gathered sufficient references in future DAG rounds that guarantee its durability. 

Upon committing an anchor node, all of its causally referenced proposals are implicitly committed and ordered.
%could also say it is "symmetric"
%\rati{complex sentence to follow}
Such design is appealingly simple -- it requires no additional communication for consensus and avoids the notoriously complex view change logic inherent to traditional BFT protocols. It is also very resource efficient -- each replica acts as a proposer, thus allowing high throughput with a large number of replicas. As a result, DAG-BFT has prompted swift adoption by several large blockchain systems~\cite{sui, aptos}. 

However, today's DAG protocols fall short in terms of latency. They require several rounds, each consisting of reliable broadcasts to \emph{certify} DAG nodes, to commit a transaction. Bullshark~\cite{bullshark} for instance requires, even in fault-free cases, up to 4 rounds to commit a proposal (12md in expectation, \S\ref{sec:dissect}); recent advances made by Shoal~\cite{spiegelman2023shoal} reduce it to 3 rounds (10.5md expected, \S\ref{sec:dissect}). A far cry from the 3md possible by traditional BFT protocols! %note: it's 9 + 1.5 for Shoal; and 10.5+1.5 for Bullshark.
%\rati{traditional BFT or just PBFT?}

%\sasha{Also, we did not yet talk about the queuing latency and the math does not make sense}\fs{added md definition early on now; the math is intentionally left out. For the intro a statement suffices. I'll add a reference to the section in which we analyze the math in detail. Note I also intentionally wrote "4 rounds to commit a proposal" and not to commit anchor}

This work proposes \sysname: a novel DAG-based BFT system that achieves near-optimal latency  while maintaining the high throughput and robustness of the state-of-the-art certified DAG construction.
%\sasha{We already have the above sentence in the first paragraph}

We begin by analyzing sources of end-to-end (e2e) latency 
in DAG protocols~(\S\ref{sec:dissect}); at a high level, it can be broken down into three stages. 
\one DAG proposals happen only at regular intervals (rounds), and transactions that narrowly miss inclusion (e.g. arrive just after a node is formed) must wait for the next round. We refer to the time a transaction must wait to be included in a DAG proposal as \textbf{Queuing Latency}.
% %Since each DAG proposal required a RB -- consisting of 3md --, ...
% \par \textbf{Queuing Latency.} The time a tFinally, the time to even be included in a proposal at all. Each DAG node is RB, so TX that narrowly miss inclusion in a batch have to wait until the next round. On avg RB/2 penalty. (but up to 3md)
\two In order to be committed, a proposed DAG node must be referenced by a committed anchor. Unfortunately, in existing DAG protocols~\cite{bullshark, spiegelman2023shoal}, at most one anchor is scheduled per round\footnote{Bullshark~\cite{bullshark} schedules an anchor every second round; Shoal~\cite{spiegelman2023shoal} improves it to every round}. This results in additional \textbf{Anchoring latency}, the time until a node is referenced by the next committed anchor.
\three Finally, in existing DAG protocols, an anchor requires \textit{at least} two more DAG rounds to be committed; we refer to this as \textbf{Anchor Commit Latency.} 
% \fs{I would move away from the term "sibling" latency; this is already implicitly assuming anchors are each round. A more general term. I've called it "Anchoring" Lat for now, but it may be easy to get confused with Anchor Lat. This is why I called it "Anchor Commit Latency" }
% \sasha{So far looks good. Let's see how it fits the next sections.}

\paragraph{Key ideas} To reduce end-to-end consensus latency, \sysname employs three techniques, respectively addressing each latency stage.
%\rati{stage called component before. below maybe: "identifies and optimizes an inefficiency in existing commit rule to reduce.." also commit rule of what, all DAG protocols, Bullshark?}
First, \sysname identifies and optimizes an inefficiency in the existing Bullshark commit rule to reduce common case Anchor Commit Latency to only 4md. 
Next, we configure \sysname to dynamically attempt to make \textit{every} node an anchor, thus eliminating Anchoring Latency. \sysname builds on Shoal's idea to dynamically re-interpret the anchor assignment at each commit, and introduces small timeouts to help replicas advance rounds in lock-step to ensure that future rounds sufficiently reference anchor candidates.
Finally, to minimize Queuing latencies, \sysname orchestrates \textit{multiple} staggered DAG instantiations in parallel, and intertwines their resulting logs; this allows transactions that miss inclusion to simply board the new round of the next, staggered DAG.

Put together, these techniques allow \sysname to reduce expected end-to-end consenus latency to only 4.5md, shaving off 6md compared to Shoal, the state of the art DAG-BFT. %In total this brings down latency to 4.5 avg: 4 for anchor + E(0) sibling + <0.5 for queuing.
%\fs{TODO: note that each node uses linear RB. if we made it all to all we could get closer to pbft}
%\sasha{Note that pbft need to 3 to commit a proposal, but another 1 to forward tnxs to the leader, making it 4 total. And on top of it there is a queuing latency that I do not know how to count for pbft but it is probably at least 0.5.}\fs{I believe "forwarding" to the leader is an artifact of replicas being proposers, and not core to pbft.}

% 1) a better commit rule for common case
% 2) A more dynamic re-assignment of anchors/fixpoints to allow "every" node to be anchor = "no more" sibling latency. This can go bad if not careful: mechanism to skip anchor + make them reliable (reputation + timeout)
% 3) Pipeline/Stagger mlutiple DAG instances to avoid wait times.

We've implemented a prototype of \sysname and evaluated it empirically against two popular DAG-BFT protocols (Bullshark~\cite{bullshark} and Shoal~\cite{spiegelman2023shoal}), a concurrent DAG-BFT protocol (Mysticeti~\cite{mysticeti}), and a traditional BFT protocols (Jolteon~\cite{jolteon}). Our results are promising, demonstrating matching throughput of existing DAG BFT solutions while reducing latency by up to 60\% over Bullshark in the common case and by up to 10x over Mysticeti in the failure case.

%% file: sections/1.5-model.tex
\section{Model}
We adopt standard BFT assumptions~\cite{castro2002practical, hotstuff}. We refer to participants that follow the protocol as correct, and those that deviate as faulty (or \textit{Byzantine}).
There are a total of $n = 3f+1$ replicas, of which at most $f$ may be faulty at any time. We make no assumptions on the number of faulty clients.
%\sasha{What clients? where did you copy this from? :)}\fs{Nowhere, it's a standard thing to have -- reviewers ALWAYS ask about this, even if its not relevant our our protocol, it is part of the standard BFT model}
We assume the presence of a strong, yet static adversary that may arbitrarily delay messages and coordinate all faulty particpiants, but may not break standard cryptographic primitives such as signatures. We denote signed messages using $\langle m \rangle_\sigma$. We assume that all signatures and Quorums are validated, and henceforth omit explicit mention.

The goal of BFT consensus is to establish a common totally ordered log across all correct replicas. We operate under the partial synchrony model~\cite{dwork1988consensus}. The system is safe at all times -- no two correct replicas ever disagree on a committed prefix of the log. Liveness is not guaranteed during periods of asynchrony, but eventually there arrives a Global Stabilization Time (GST) after which the network (temporarily) behaves synchronously, and progress is guaranteed.

\iffalse
The goal of a BFT system is to allow a set of validators to continuously propose transactions and to agree on an ever-growing ordered log of them. We consider a traditional setting of a peer-to-peer message-passing network comprising $n$ validators, with a computationally bounded adversary. We design \sysname to be secure under worst-case assumptions while optimizing it for the common case. \sysname guarantees that validators never disagree on prefixes of the log (Safety) against an adversary that has full control over the network and can corrupt up to $f < n/3$ \emph{Byzantine} validators. Theoretical liveness, ensuring that the log eventually grows, is guaranteed in partially synchronous networks in which there is a global stabilization time after which messages arrive within a bounded delay. The practical latency of \sysname for transactions to be included in the ordered log, however, is optimized for the crash-failure case, in which validators are not corrupted but may crash or be extremely slow. 
\fi

%% file: sections/2-background-new.tex
\vspace{-2mm}
\section{Dissecting DAG-BFT}\label{sec:background}
%Towards Low latency Dags
Before we present \sysname, we provide a brief background on the core mechanisms of DAG-BFT, and analyze the latency breakdown of existing protocols. We implicitly focus on \textit{ certified} DAG protocols derived from Narwahl\cite{narwhaltusk}. Although some recent works propose uncertified DAGs in an effort to save latency, this comes at the cost of reduced robustness and introduces new sources of unwanted latency (\S\ref{sec:dissect}).

%1) Give background, 2) break down lat, 3) give key ideas
%\fs{TODO: rewrite!}

%\subsection{Background}

\vspace{-2mm}
\subsection{DAG core}
At the core of DAG-BFT is the round-based certified DAG structure proposed in Narwahl~\cite{narwhaltusk}. Replicas proceed through a series of structured rounds, in which replicas continuously propose new nodes containing batches of transactions. 
Each replica maintains their own local view of a DAG.

A replica adds another replica's proposal as node in its local DAG once the proposal is \textit{certified}. Proposals in round $r$ must reference $n-f$ certified proposals from round $r-1$ to be considered \textit{valid} candidates, and each replica may propose and certify at most one candidate per round. %The resulting DAG consequently contains at most $n$ nodes per round. 
%\balaji{what do you mean by "each replica may certify at most one candidate per round". It seems to mean only one node is certified per round?}
To certify a proposal for round a replica follows a simple Reliable Broadcast procedure. % (Figure~\ref{fig:narwhalround}): \fs{Fig 2 is visually quite confusing. Maybe just draw a SINGLE certification process? I.e. just one Node zoomed in? }
\begin{enumerate}
    \item It broadcasts a signed proposal $P \coloneqq \langle r, B, edges \rangle_\sigma$ for round $r$ containing \one a batch $B$ of transactions, and \two $n-f$ unique proposal certificates from round $r-1$ (\textit{edges} for short). 
    \item Upon receiving a valid proposal, a replica checks whether it has already received a different proposal from the same author for round $r$. If it has, it ignores the proposal; if it has not, it casts a signed vote $S \coloneqq \langle r, d=hash(P) \rangle_\sigma$ 
    %\sd{check $\sigma$ here.} 
    back to the proposing replica.
    \item A replica waits for $n-f$ matching signatures for its proposal $P$ and aggregates them into a certificate $C \coloneqq \langle r, d, \{S \} \rangle $ that it broadcasts to all replicas.
    \item Upon receiving a proposal certificate a replica adds a node to its local DAG.
\end{enumerate}

%\sasha{Changed vote (V) to signature (S).}
%\sasha{Note that if the DAG is not certified you cannot do that without having all the causal history.}\fs{Ack: I made a section on Uncertified further down; but if you want you can mention it here.}

% \begin{figure}[h]
%     \centering
%     \includegraphics[width=0.4\textwidth]{figures/narwhalround.pdf}
%     \caption{One round on the DAG from a validator's local point of view. The validator receives uncertified nodes pointing to $n-f$ certified nodes in the previous round, signs them, sends back the signatures, and later gets the nodes' certificates.}
    
%     \label{fig:narwhalround}
% \end{figure}

%\fs{Add a sentence: This protocol guarantees that all replicas eventually see the same DAG. (i.e. it is equivocation free and each node is available for sync). Then for consensus its easy to say: "since DAGs are the same, its enough to interpret the local view -> no extra messages need to be exchanged"}

Figure~\ref{fig:narwhaldag} illustrates a DAG of certified nodes as it appears in a local view of one of the validators. 
We note that because proposals are certified, no replicas can produce two conflicting nodes for the same (round, replica) position. Thus, all replicas eventually converge on the \textbf{same} local view of the DAG. %This property is crucial to facilitate embedding consensus logic directly into the DAG (\S\ref{sec:embedded}).
%\sasha{Is it crucial? Cordial miners also embed consensus logic into the DAG, no?}\fs{I guess; I've removed the sentence.}

% The advantage of certified nodes is twofold: first, it eliminates the ability of Byzantine validators to equivocate and thus simplifies the DAG ordering logic, and second, it provides Poof of Availability~\cite{cohen2023proof, quorumstore} for the ancestors' nodes, which makes DAG fetching and data streaming much simpler (see more discussion in Section~\ref{sec:practical}). 
% \fs{it also makes it more robust in a principled way (see Autobahn)}
 
\begin{figure}[h]
    \centering
    \includegraphics[width=0.35\textwidth]{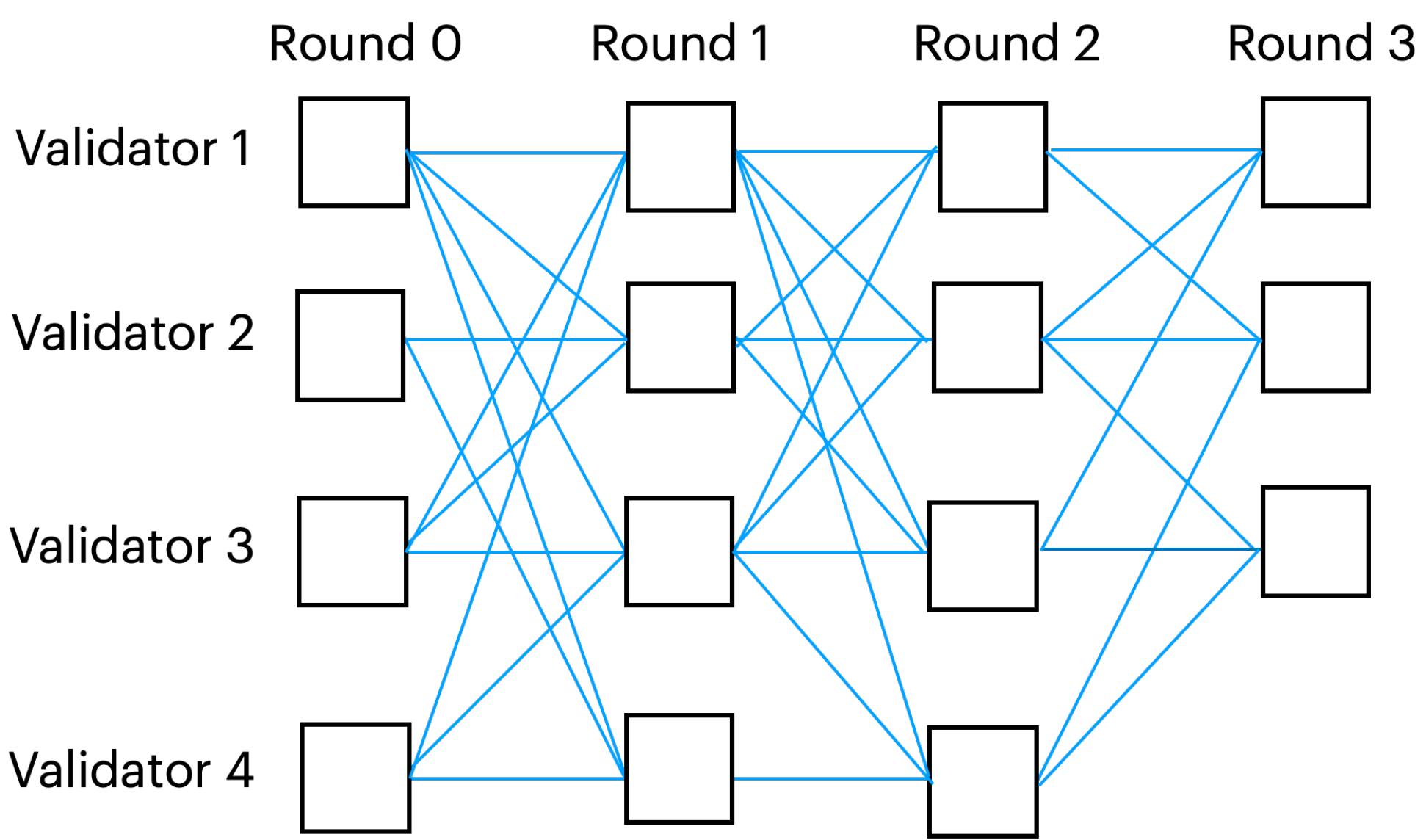}
    \caption{Narwhal's round-based DAG from a validator's local point of view.}
    
    \label{fig:narwhaldag}
    \vspace{-8pt}
\end{figure}

% \vspace{-3mm}
\subsubsection{Consensus core.}\label{sec:embedded}  
By itself, the DAG is no more than a scalable and robust mempool. It continues to grow at the pace of the network (i.e. \textit{responsively}), and, by virtue of being certified, ensures that all disseminated transaction batches are reliably available.
A consensus layer is needed to establish agreement. It can be modularly layered by running an \textbf{external} consensus blackbox that commits views of the DAG (e.g., Narwhal+Hotstuff~\cite{narwhaltusk}), or logically \textbf{embedded} into the existing DAG messages without the need for additional messages (e.g., Tusk~\cite{narwhaltusk}, Bullshark~\cite{bullshark}, or Shoal~\cite{spiegelman2023shoal}). The rest of the paper focuses on embedded DAG-BFT protocols as they are more efficient.

\paragraph{Embedded consensus}
At a high level, consensus is projected onto the DAG structure by designating specific DAG nodes as fixpoints (henceforth \textit{anchors}) that simulate a leader, and interpreting DAG edges as votes for a leader. %\sasha{anchors or leader? need to be consistent. I prefer anchors.}\fs{yes, anchor from here on out; just for the intuition leader} 
Upon committing an anchor, all of its causal history can be implicitly committed as well; an ordering can be derived using any deterministic function, e.g. using a topological sort.

Anchor nodes are chosen deterministically in advance (much like leaders in traditional BFT protocols) and placed at regular intervals. Bullshark~\cite{bullshark}, for instance, uses a round-robin scheme to select anchor candidates and places an anchor candidate per every other round of the DAG. Figure~\ref{fig:bullshark} illustrates the anchor commit logic in Bullshark.

\begin{figure}[h]
    \centering
    \includegraphics[width=0.4\textwidth]{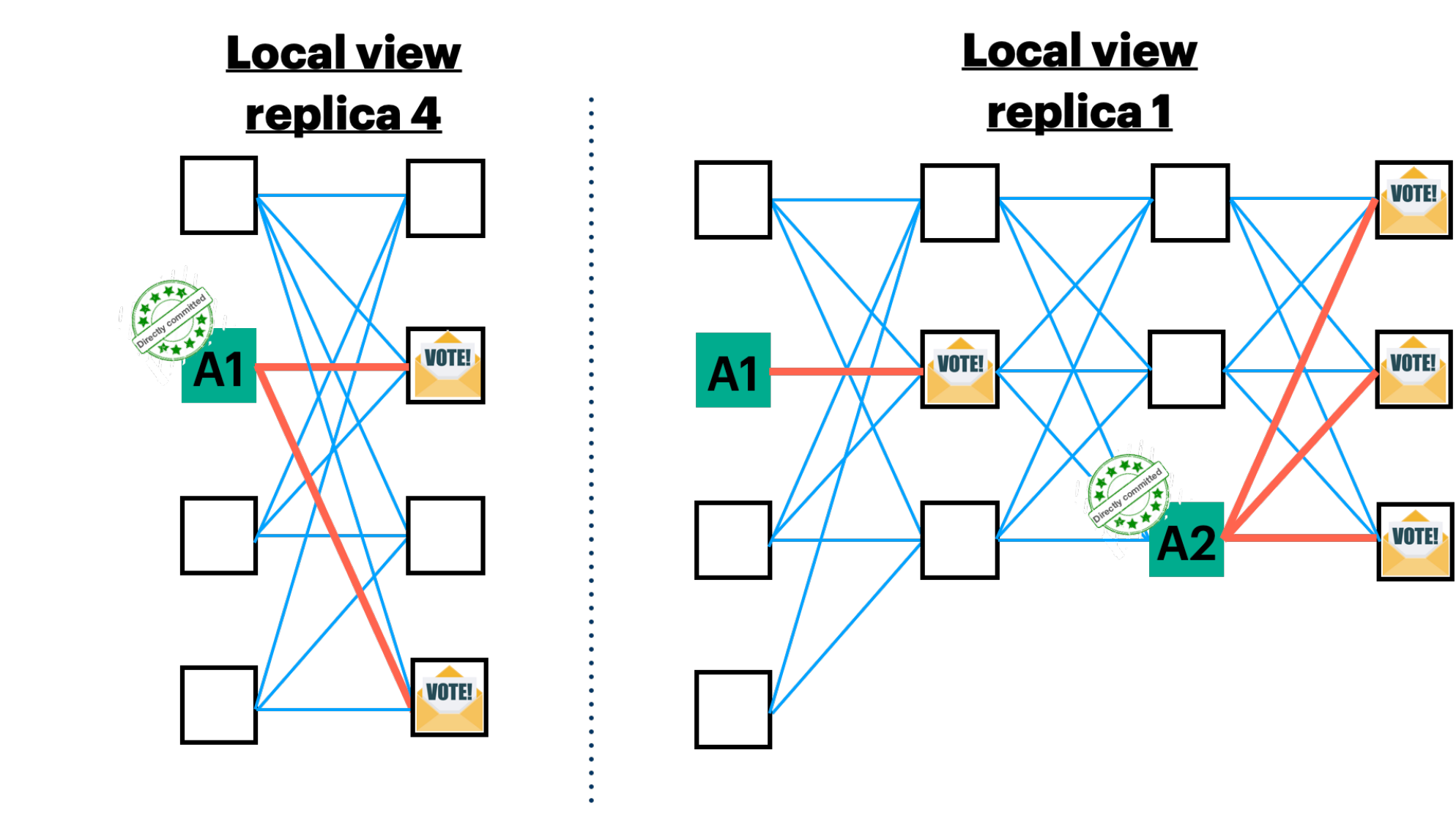}
    \caption{Bullshark Commit example. Replica 4 observes $f+1 = 2$ votes for anchor A1, and directly commits it. Replica 1 observes only one vote ($< f+1$), and cannot directly commit A1; it commits A1 indirectly upon committing A2 and its causal history, ordering A1 before A2.}
    \label{fig:bullshark}
    \vspace{-12pt}
\end{figure}
An anchor in round $r$ becomes committed as soon as $f+1$ nodes in a round $r' > r$ link to it (\textbf{Direct Commit Rule}). The safety intuition is straightforward: \one all nodes are certified, and thus every anchor is equivocation-free, and \two all nodes in rounds $r'' > r'$ must include an edge to at least one of the nodes that link to the anchor, and thus the anchor is durable.

In particular, every future committed anchor must link to all previously committed anchors, thereby establishing a consistent total order amongst anchors. Upon committing an anchor, a replica recursively traverses the anchor's causal history, and first commits all preceding anchors (\textbf{Indirect Commit Rule}). Conversely, if an anchor candidate is \textit{not} in a committed anchors causal history, then it could never have committed, and it is safe to skip it.
Intuitively, the Indirect Commit Rule ensures that, even if a replica does not invoke the Direct Commit rule on some anchors (e.g. because its local DAG view is temporarily out of sync), all correct replicas will agree on the same sequence of committed anchors.
%\sasha{Do we need this correctness argument? No strong opinions, but it is not a must for understanding the contributions. Well summarized though.}
%\sasha{I would keep it for now, but this is a candidate to save space if needed.} \fs{I think it helps with the intuition a lot.}

We note that while this design is always safe, it relies on conservative timeouts to ensure liveness (similarly to traditional BFT protocols).
Since nodes include only $2f+1$ edges, anchor candidates are not guaranteed to receive the $f+1$ links required to commit. In the worst-case, candidates are repeatedly skipped, thus stalling commit progress.

To reduce the expected latency, Shoal~\cite{spiegelman2023shoal} optimizes Bullshark in two ways. First, it leverages the deterministic anchor commit sequence to dynamically re-interpret the anchor schedule upon every commit. Rather than placing anchors every other round, this approach allows for an anchor to be placed every round.
Second, it introduces an anchor reputation scheme to attempt to select as candidates only the fastest and well-connected replicas; such candidates are more likely to be included as edges. This allows Shoal to avoid timeouts in all but extremely unlikely scenarios.
%\sasha{I edited the paragraph. Also, when are we going to mention multiple anchors idea per round was proposed in Shoal?}\fs{looks good. And: I've mentioned it in a footnote in the technical section S4.}

\iffalse
- Key highlights:
    - leader every other round
        - Shoal makes this every round.
    - committed with only f+1 rounds
    - 2 dag rounds to commit in common case;
    - bad case can be longer (higher anchoring latency)
        - Shoal uses leader reputation to make this less likely.

- subtle: if an anchor is committed LATE => then high anchor latency. If an anchor is skipped => then high anchoring latency.
   - in a way, the anchor that was skipped had to pay its own anchoring latency.

Bonus: Shoal avoids timeouts -- we can do the same here; it's not really that important here though. Too much detail
\fi

\vspace{-2mm}
\subsection{Latency breakdown}\label{sec:dissect}
%While the design of embedded DAG-BFT protocols is appealingly simple, existing protocols pay a high latency premium. 
In the following, we breakdown different sources of latencies present in DAG-protocols, and analyze their respective costs in Bullshark and Shoal.
%\balaji{here we use commit latency. then we have transaction latency and consensus latency in intro}
%
We denote as the e2e consensus latency of a transaction the time between a replica first receives the transaction, and the time that the transaction is ordered. We separate the e2e latency of certified DAG-BFT protocols into three components:
\begin{enumerate}
    \item \textbf{Queuing latency:} the time it takes for a transaction to be proposed in a DAG node.
    \item \textbf{Anchoring latency:} the time it takes for a non-anchor node's proposal to be "picked up" by an eventually committing anchor.
    \item \textbf{Anchor Commit latency:} the time it takes to commit an anchor proposal.
\end{enumerate}

%\fs{Maybe have a table? (Narwahl, Bullshark, Shoal, us, Mysticeti?). Do you want to open the door to uncertified DAGs here too? I think its probably confusing at this point; and should be left to related work}

%\fs{note: I'm trying to use "best case average" instead of expected; we cannot quantify the expected with pen and paper; it depends a lot on the quality of repuation and network behavior. Is the difference clear enough? or does it need more explanation}
%\sasha{Expected is only for the queuing latency assuming txns arrive at a uniform distribution. What you did below is good.}

\paragraph{Queuing Latency} 
The DAG advances in a series of synchronized rounds, and each replica may issue at most one proposal (containing a batch of transactions) per round. Each round requires a Reliable Broadcast to certify a proposal, for a total of 3md. Consequently, a transaction that narrowly misses inclusion in the current round proposal must, in the worst-case wait up to 3md before being broadcast; assuming transactions arrive at a uniform rate, the average Queuing Latency is 1.5md.
%\sasha{I think before I explained earlier that we do not focus on non-certified DAGs. Now I am afraid people can attack us 3md is specific to Narwhal-based DAG but we present it as a more general thing.}
%\sasha{Ok, I see that you mention it on top and have a discussion later. I will think about what I think about it.}\fs{ack.}

\paragraph{Anchoring latency}
A proposal becomes "anchored" once it is referenced by an eventually committing anchor. Anchoring latency depends on two factors: \one the frequency of anchor candidates, and \two the reliability with which said anchor candidates are committed (i.e. not skipped).\footnote{The proposals of non-anchor candidates are subject to both factors; the proposals of anchors candidates themselves, trivially, are subject only to the latter.} 
Bullshark has anchor candidates in every odd numbered round (i.e., only every second round). Consequently, the best-case inclusion latency of non-anchor proposals in even rounds is 1 round (3md), and 2 rounds (6md) in odd rounds; a best-case average of 4.5md. Shoal increases candidate frequency to every round, and thus improves the inclusion latency of non-anchor proposals to one round for all  (3md).
% approx bound for direct commit: need at least f+1 correct replicas, each has p ~= 2/3rd chance to pick anchor.
%N = 3f+1; M = 1. n = 2f+1; x = 1
%f(x) = B(N-M, n-x)/ B(N, n) = B(3f, 2f)/B(3f+1, 2f+1) = prob to pick the anchor.
% f= 1 => 4/6 = 2/3
% f= 10 =>  30045015/ 44352165 = ca 2/3.
% 2/3 prob that any given replica picks the leader: need at least f+1 picks. (out of 2f+1; can't rely on byz f). So need at least half replicas to vote for leader: 1-  sum(x: 0 to f): b(f, x) (2/3)^x * (1/3)^(2f+1-x)  (f+1, or f+2.. up to f times)
% => TODO is there an easy closed form?
% for f=1: 
% 2f+1 = 3: need at least 2 picks: so (2/3)^3 + 3* (2/3)^2 * (1/3) = 8/27 + 12/27 = 20/27 ~= 74%
% for f=2: 2f+1 =5; need at least 3 picks. 
% (2/3)^5 + 5* (2/3)^4 * 1/3 + 10 * (2/3)^3 * (1/3)^2 = ca 79%
%\footnote{In theory, if networking was truly symmetric and synchronous, the probability that an anchor commits using the direct commit rule could be quanitified as }
% \fs{todo: in theory could put some math here to quantify how reliable it is (I put it in comments). Don't think it's particularly useful though.}
The practical reliability of commitment cannot be easily quantified by pen-and-paper, and is best left to empiric analysis as it depends on network configuration and message delivery patterns. Shoal demonstrates that the use of a reputation can aid significantly in ensuring that anchor candidates are reliably included~\cite{spiegelman2023shoal}.

%=> Anchoring latency in Bullshark = 2 rounds; in shoal = 1. This is in the "good" case, I.e. the next anchor succeeds. If the anchor dosn't succeed, then it's + whatever gets missed. => I.e. the latency to the anchor that commits us.
%Shoal helps with sibling latency 

\paragraph{Anchor Commit latency} In both Bullshark and Shoal, an anchor becomes committed upon either \one being references by $f+1$ nodes in the next round (Direct Commit Rule), \two or being subsumed in the causal history of future committed anchor (Indirect Commit Rule).
%\sasha{nit: "being referenced" above is a path, but text may imply a single edge.}\fs{I changed it to: subsumed in causal history.}
The Direct Commit Rule requires at least 2 rounds: one to certify the anchor proposal, and one to certify $f+1$ proposals that reference it, for a total of at least 6md. The latency of indirectly committed anchors can be broken down into their own Anchoring Latency to a future anchor, plus the future anchors Commit Latency.

Putting all latency stages together results in a best case average latency of 12md (1.5md + 4.5md + 6md) for Bullshark, and 10.5 for Shoal (1.5md + 3md + 6md). 
%The overall expected latency in Bullshark may be significantly higher as anchor reliability
%\sasha{I would delete the last sentence since we do not improve the worst case in this work.}

\subsection{A note on uncertified DAGs} A class of recent DAG-based designs~\cite{keidar2022cordial, mysticeti, bbca} proposes to replace DAG nodes with uncertified proposals that are disseminated with best-effort. At first glance, such designs holds promise by reducing latencies across the stack. However, removing certificates makes the DAG brittle: it can introduce a significant amount of unwanted synchronization on the critical path and increase susceptibility to timeout violations~\cite{giridharan2024autobahn}. Edges to uncertified, yet locally unavailable proposals must be validated by fetching missing data (at least 2 additional md per round) \textit{before} advancing to avoid losing liveness in the face of Byzantine fabrications or unstable networks. In practice, we find that even occasional message drops can significantly degrade overall latency (\S\ref{sec:failures}). In contrast, \sysname demonstrates that certification is \textit{not} the root-cause of high latencies, and that latency can be significantly reduced without sacrificing robustness.
%\sasha{That is a good paragraph. I wonder if we should also say that fetching becomes a bottleneck because we do not know where to fetch from. In certified DAGs, in contrast, we had all the replicas that signed the node to fetch for its parents. I had this somewhere, but cannot remember where.}\fs{This is something for the practical consideration section -> there we can mention the load balanced fetch since we know the bitmap.}
%\sasha{Do we want to have here a discussion about the fetching issue in uncertified DAGs?}

\iffalse
=> all have anchor latency = 2 rounds. direct commit
=> indirect = anchoring + direct commit.
- Note: if not certified => +2 md to sync...
\fi

%\vspace{-1mm}
\section{\sysname Overview}

%\sasha{Should we have a separate section for this?}\fs{Either is fine with me. I felt it fit well here.}
%\sasha{Not sure it fits the section title "Dissecting DAG-BFT"}

\sysname strives to reduce latency across each of the three latency stages; it employs three different ideas that respectively address Queuing, Anchoring, and Anchor Commit Latency. Together, \sysname is able to reduce total e2e expected latency close to an average of 4.5md. We briefly outline each idea in sequence.

%Goal: Try to address each latency part:
%1) A faster direct commit rule for the common case
\paragraph{Faster anchors} The first observation \sysname makes is that the Bullshark Direct Commit rule can be optimized when anchors are referenced by a supermajority of uncertified nodes.
%\sasha{What does reliably included mean?}
Intuitively, a replica may eagerly commit an anchor as soon as $2f+1$ \textit{proposals} that link to an anchor node are observed, as this implicitly guarantees that eventually the Direct Commit Rule is satisfied. Since this is an even stricter condition that may not be fulfilled, \sysname \one retains the existing Direct Commit Rule as backup, and \two leverages Shoal's reputation mechanism to reliably trigger the Fast Direct Commit Rule. This allows \sysname, in practice, to reduce Anchor Commit latency to 4md.
%\sasha{this paragraph needs some work. Maybe it is my original writing, but it is not clear enough. What does it mean to make it common? }\fs{I made some small edits based on your comments. What do you think now.}
%\sasha{Better.}

%This helps if anchors are reliably included; we adopt Shoals reputation, and try to add small timeouts to make it reliable. We keep the old rule; in best case 4md!

%TODO: just say waht is new and then say why we need to make it robust?

\paragraph{More anchors}
Next, \sysname takes Shoal's idea of increased anchor frequency one step further, and tries to turn as many nodes into anchors as possible. Doing so requires striking balance. On the one hand, increasing the number of anchors increases the opportunities to avoid incurring the anchoring latency. On the other hand, it introduces also opportunity for more uncommitted anchors, which can stall progress. Adopting the reputation scheme proposed by Shoal can help identify a reliable subset of candidates, but is not enough to streamline progress in practice. \sysname \one slows down round advancement to put the DAG into lockstep and avoid missed anchors, and \two dynamically re-interprets the tentative anchor schedule to skip anchors that are no longer needed.
% further \one introduces a small round timeout, which altruistically increases latency slightly locally,

%Question -- actually eveyr node anchor? or only those within reputation..? 
%Use timeout to slow down everything a tad, to in practice make overall better

%2) Virtually eliminating ancohring latency in the common case
%Note: Anchoring latency is variable. We want to make sure that we make it low and reliable.

\paragraph{More DAGs} Finally, \sysname minimizes Queuing latency by operating not one, but multiple DAGs in parallel. DAGs are staggered using a small offset, allowing waiting transactions that miss a round to quickly be included in the \textit{next} DAG. This simultaneously improves throughput via better resource utilization: proposals are sent more frequently, batches are smaller and can be sent and processed more efficiently (in a streaming fashion) than less frequent, large batches.
%\sasha{"allowing waiting transactions to miss a round to quickly be included" is grammar correct here?} \fs{fixed "to miss" => "that miss" }

% \section{\sysname Details}
% \fs{THEN: Full Shoal}
% \section{Correctness}
% \fs{then proof (if any)}
% \section{Practical Considerations}
% \fs{then practical considerations (short)}

%% file: sections/3-protocol.tex
\section{\sysname In-Depth}
%We now describe the full \sysname system in detail. 
\sysname uses as starting points the DAG-core of Narwahl~\cite{narwhaltusk}, and Consensus-core of Bullshark~\cite{bullshark} outlined in \S\ref{sec:background}. For the sake of brevity, we do not duplicate the pseudocode that is available in previous works and, in particular, omit a detailed description of the DAG construction (i.e., node broadcasts, transaction submission, node fetching, garbage collection, and so on). We defer to the original papers for an in-depth discussion and correctness proofs.

Algorithm~\ref{alg:dagapi} summarizes the DAG API and node structure. We abstract away the node certification process, and the commit logic. \textsc{Process\_node} and \textsc{Process\_certified\_node} respectively involve checking that a node proposal is valid and voting for it, and adding a new confirmed node to the DAG, triggering commit rules. \textsc{Run\_bullshark} performs an instance of the Bullshark commit logic to determine (given a starting point) the first anchor to be ordered (\S\ref{sec:embedded}).
%\sasha{Changed this to be aligned with the new pseudocode, I hope the "first anchor" is not confusing.}\fs{changed it a tad more, but looks good}
%, which, given an anchor, interprets the DAG as a binary consensus instance. It returns true if the anchor is to be committed, and false if it is skipped.

Upon committing an anchor, all of its (yet-to-be-ordered) causal history is deterministically ordered as well.
Helper functions \textsc{causal\_history} and \textsc{get\_anchors} aid, respectively, in identifying the causal history of a node (necessary for both commit and ordering logic), and the anchor candidate schedule. \sysname adopts and extends Shoal's~\cite{spiegelman2023shoal} reputation logic; \textsc{get\_anchors} returns, for each round, a vector of eligible anchors.

\input{code/dagapi}

\input{code/mako}

Algorithm~\ref{alg:mako} summarizes the core protocol additions made by \sysname in an effort to reduce Anchor Commit, and Anchoring latency. We now discuss both in detail in turn.   

\subsection{A faster Direct Commit rule}

As described in \S\ref{sec:dissect} it currently takes, in the best case, 6 message delays to commit an anchor using the Bullshark Direct Commit Rule. One DAG round is required to certify the anchor node itself, and a second DAG round is required to collect sufficient votes. A replica can directly commit an anchor upon observing $f+1$ certified nodes that link to the anchor (i.e. include the anchor in their parent set).

\begin{figure}[h]
    \centering
    \includegraphics[width=0.33\textwidth]{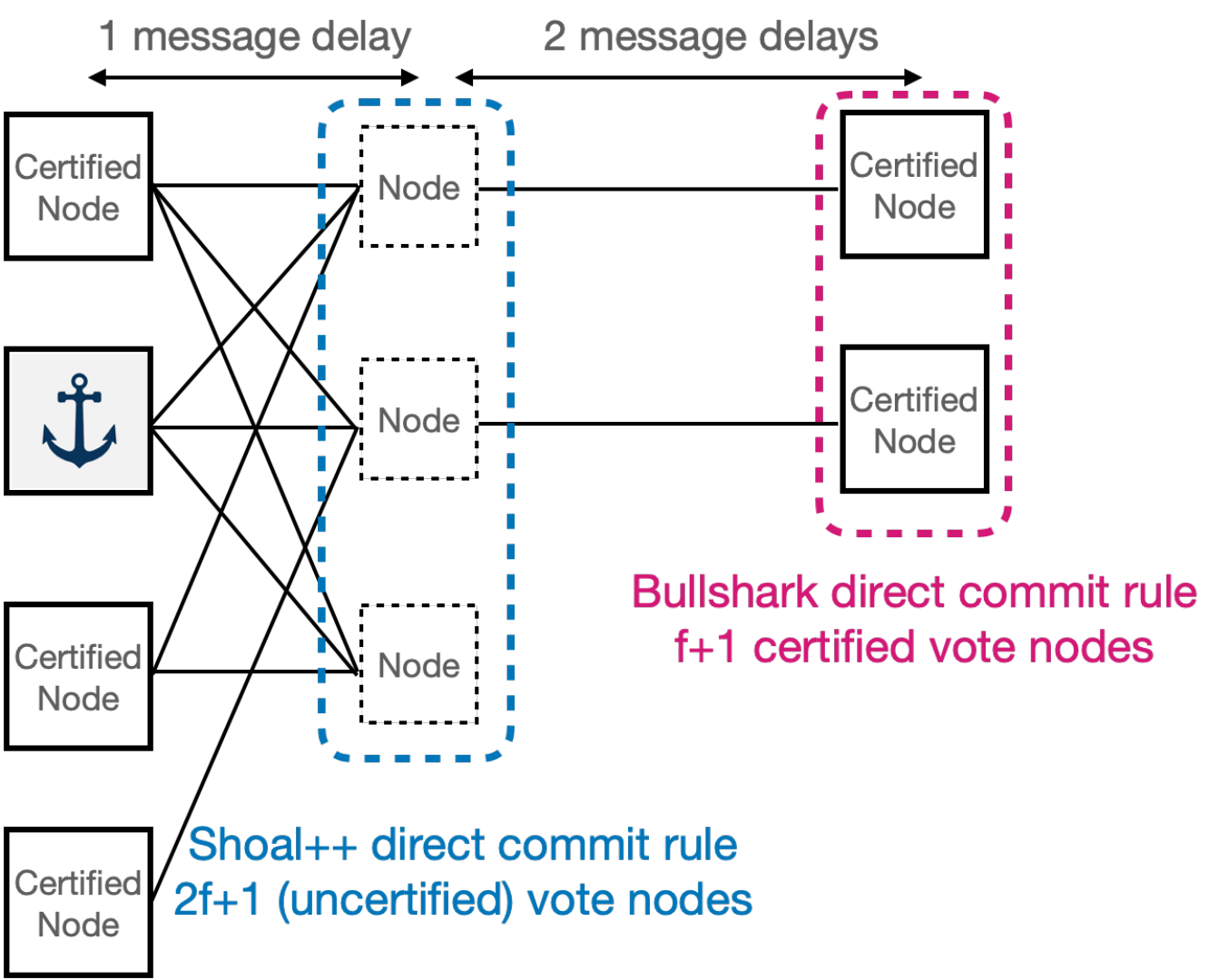}
    \caption{Bullsharks's direct commit rule requires f+1 certified vote nodes, while \sysname's direct commit rule requires $2f+1$ (uncertified) node proposals.}
    \label{fig:commitrule}
\vspace{-15pt}
\end{figure}

\sysname makes the following simple observation: if sufficiently many DAG \textit{proposals} link to the anchor, the anchors fate is already set in stone, and we need not wait for round certification to complete. Recall that the certification process consists of three message exchanges, a proposal step, a voting step, and a certificate forwarding step (\S\ref{sec:background}). 
\sysname allows replicas to \textsc{Fast\_Commit} (Alg. \ref{alg:mako}) as soon as it observes $2f+1$ proposals that link to an anchor (Fig. \ref{fig:commitrule}). To do so, replicas simply additionally keep track of \textit{weak votes} for each round that correspond to the proposals observed.
Weak votes are not guaranteed to survive certification. For instance, a faulty proposer may equivocate, and produce a node certificate containing a different proposal that does not link to the anchor. However, out of any $2f+1$ weak-votes, at least $f+1$ must be cast by correct proposers that will never equivocate. Consequently, the presence of $2f+1$ weak-votes guarantees that eventually $f+1$ certificates linking to the anchor must be formed, thus ensuring safety. 
Using the Direct Fast Commit rule allows replicas to commit an anchor in only $4$md: 3md to certify the anchor, and 1md to broadcast and receive proposals.

We note, however, that this rule places a \textit{stronger} requirement for liveness; $f$ additional proposals are (tentatively) required to link to an anchor. 
We thus retain the existing Direct Commit rule as backup, and allow replicas to commit using whichever rule is satisfied first. In some cases, it may be faster (and more reliable) to commit using $f+1$ certified nodes. For instance, when the network is unstable or point latencies are asymmetric, the fastest $f+1$ replicas might advance much more quickly than the fastest $2f+1$.

In our empirical evaluation, we find that this case is rare (\S\ref{sec:evaluation}) and that replicas are, most of the time, able take advantage of the Fast Direct Commit rule. Shoal's reputation scheme ensures that eligible anchor candidates are typically well connected, making it exceedingly commonplace for future round proposals to link to an anchor.

\subsection{Increasing anchor frequency}
%\fs{todo say that shoal suggests doing this, but doesn't eval it. }
%\sasha{re footnote: I would say it is even more than that. Shoal paper states that in theory it can be done, but in practice would probably result in worse latency because one slow anchor forces increased latency on all the following in its round. }
In order to reduce Anchoring latency, \sysname tries to dynamically designate as many nodes as possible as anchors. Intuitively, if all nodes were anchors that become committed, no node would experience any Anchoring latency.\footnote{We note that such an approach is briefly discussed in the closing remarks of Shoal~\cite{spiegelman2023shoal}, but is neither implemented nor carefully evaluated for practical considerations.}
Safely operating more than one anchor per round requires additional orchestration. To ensure a consistently ordered log across all correct replicas, parallel anchors must be committed in a pre-defined order. One easy way to visualize parallel orchestration is by mapping anchors to fixed slots; much like parallel proposals in traditional BFT protocols such as PBFT~\cite{castro2002practical}. Figure \ref{fig:multipleanchors} shows a simple example in which there are four anchors in one round (all $n$ nodes are anchors). Each anchor employs it's \textit{own} consensus instance that may proceed in parallel, and out of order, but anchors can be committed only in sequence. Each consensus instance operates a one-shot Bullshark instance, augmented with the Fast Direct Commit Rule. We remark that each Bullshark instance operates it's own tentative anchor schedule (e.g. the yellow anchor in Fig. \ref{fig:multipleanchors} is considered the next anchor in the instance started by the red anchor), but terminates as soon as the initial anchor is resolved.

\begin{figure}[h]
\vspace{-3pt}
    \centering
    \includegraphics[width=0.3\textwidth]{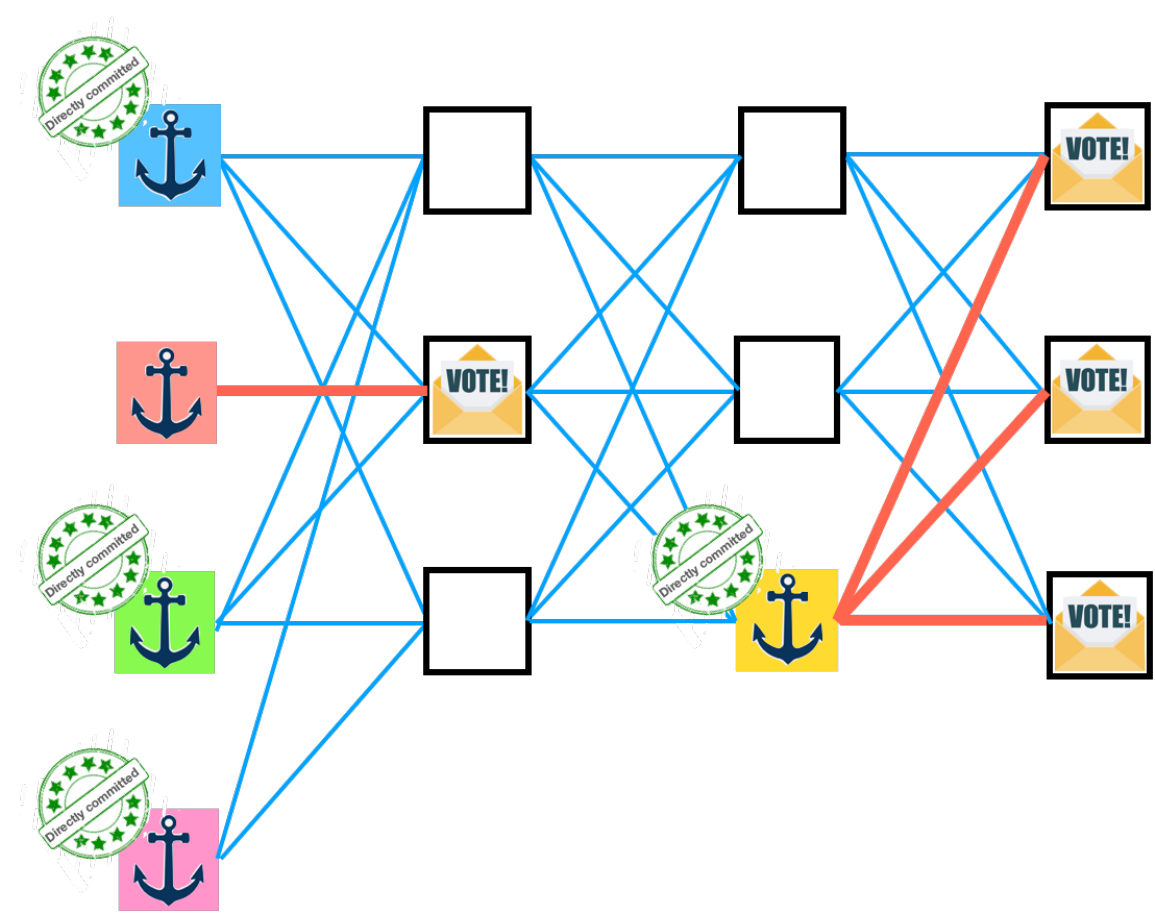}
    \caption{An illustration with four anchors in round $1$, using the pre-assigned order: blue, red, green, pink. The blue anchor commits directly; it has $\geq f+1$ references from nodes in round $2$. The red anchor is only indirectly committed upon directly committing the yellow anchor in round $4$. The green and pink anchors must wait to be ordered until the red anchor is resolved.}
    % \caption{All nodes in round $1$ are considered anchors, and each is resolved (committed or skipped) using an independent consensus instance. Anchors are implicitly ordered a-priori, and their commitments can only be applied in sequence. In the best case, all anchors are directly committed in round $2$. In the bad case, some anchor lacks support, and is only resolved later (via the Indirect Rule). In this example, the blue anchor is considered first and ordered fast since it can be directly committed with $3 \geq f+1$ votes in its instance. However, the green and the pink anchors must wait two extra rounds for the orange to be (indirectly) committed even though they are committed directly in their Bullshark instances.}
    \label{fig:multipleanchors}
    \vspace{-10pt}
\end{figure}

While simple, a straightforward instantiation of this design is not robust, and may, in practice negate (or even hurt) any latency benefit. By treating every node as an anchor, progress is, in particular, bottlenecked on the \textit{slowest} anchor. We illustrate such an unwelcome scenario in Figure \ref{fig:multipleanchors}: although the green and pink anchors have sufficient support to commit, they must wait until the preceding red anchor becomes committed or skipped. Unfortunately, the red anchor is slow and does not have any immediate links; it is indirectly committed only two rounds later via the next anchor in it's consensus instance (yellow anchor).
In worse circumstances, an anchor might never commit. In this case, we must wait for confirmation that the anchor is skipped, and fill its assigned slot with a no-op. 
Such scenarios are not fringe outliers, but expected by design: since DAG nodes include only $2f+1$ edges, rounds are privy to advancing based on the fastest $2f+1$ nodes, and leaving the remainder behind. Thus, with high probability, \textit{at least} one anchor per round will not receive sufficient direct votes.

Adopting the reputation mechanism proposed by Shoal can help ameliorate, but not remedy the concern. We leverage reputation and modify \textsc{Get\_anchors}(r) to limit the set of eligible anchors per round only to those associated with historically fast and well-connected replicas. 
%In this regime, the remainder of round nodes ($f$) must remain non-anchors and are committed once anchored, i.e. as part of the causal history of a committing anchor.
% \sasha{We actually try to have as many anchors as possible. This is why we use the timeouts.} \fs{I say this in the next section.}
Unfortunately, we find that in practice leveraging reputation alone is insufficient, and does not reliably enough exclude slow replicas. While past performance is a helpful indicator, network instabilities in geo-distributed systems make it hard to accurately predict \textit{future} performance, and we observe that the set of "slow" replicas changes very dynamically. To improve robustness further, we add two additional techniques that, respectively, try to \one avoid unnecessarily classifying nodes as slow, and \two dynamically skip likely obsolete anchors.

\paragraph{Round Timeouts} Reputation alone is overly harsh in categorizing nodes as slow. Since DAGs can advances with the fastest $2f+1$ replicas, to maximize the likelihood of anchors that commit directly, reputation \textit{pessimistically} caps the eligible anchors per round to include only the (estimated) fastest $2f+1$ nodes. (i.e. $\frac{2}{3}$ of total nodes).
%\sasha{I do not get this $\frac{2}{3}$ part.}\fs{this is what you wrote before? 67 percent?}
In reality, however, ineligible nodes may only be fractionally slower, and arrive imminently. To avoid discriminating against such nodes, \sysname makes the conscious choice to, in each round, additionally wait for a small timeout beyond the first $2f+1$ nodes observed, and include any momentarily received nodes as additional edges. While this slightly slows down individual rounds, it encourages replicas to operate in lockstep, resulting in a more densely connected DAG. This allows us, in turn, to optimistically increase the set of eligible anchors back to include all $n$ nodes, thus reducing \textit{overall} latency. 
 
We note that timeouts are not necessary for liveness, and are purely optimizing performance. They can thus be configured very aggressively, and are, in practice, negligible compared to the 3 network delays required to certify a node.

% \fs{TODO: say something like "he skip case is the more likely one", because intuitively the round timeout already makes it so that those that are connected use Direct Commit rule instead of Indirect.}
\paragraph{Skipping Anchor Candidates} Although round timeouts increase the reliability of the Direct Commit Rule, they cannot guarantee its success. Occasionally, replicas may fail or experience high networking delays, resulting in missing, or unconnected nodes. In such cases, all subsequently pre-ordered anchors must wait until the current anchor is resolved. Recall that an anchor $A1$ in round $r$ is only resolved only upon committing, in its local consensus instance, an anchor $A2$ in a future round $r' > r$: $A1$ is committed if it is present in the causal history of $A2$ (Indirect Commit), or skipped if it is not. In practice, we expect indirectly committed anchors to be rare. Anchors that are merely slow, are likely to be committed thanks to the round timeout. Crashes, in turn, predominantly result in skips. Byzantine abuse
%, such as disseminating certificates to only a subset replicas, 
can unfortunately not be avoided, 
and we must rely on the reputation scheme to avoid such anchors; fortunately, such faults are exceedingly rare in most practical deployments.\footnote{Blockchain systems~\cite{sui, aptos} report observing zero Byzantine faults in over a year of deployment. Reported faults are typically benign crashes.} % (unless they failed mid broadcast).
% \sasha{Not really mid, more at the last stage of sending the certificates.}\fs{No: I mean "MID broadcast" as in = send to some nodes but not all. But if you think it's needlessly confusing, just remove the part in parentheses)}
% \sasha{I get what you are saying. I am just saying the broadcast has two parts. If you crash in the first (casting the node) it is like crashing before starting.}
%\fs{I've removed the mid broadcast thing. Instead, I added a discussion of Byzantine faults.}

We note that, even though $A2$ has been committed, it cannot be ordered before all preceding anchors in the pre-defined schedule have resolved their respective consensus instances (i.e. commit or skip). This is undesirable: such anchors are \one likely already included in the causal history of $A2$, and \two may be subject to further latency delays. Unfortunately we cannot avoid this scenario if $A1$ is committed; another replica may have committed it directly, and resolved (and applied) the output of the next scheduled anchor $B1$. 
We can, however, sidestep the resolution of preceeding anchors' Bullshark instances if $A1$ is skipped. For any given consensus instance, all replicas necessarily agree on the order of \textit{all} committed and skipped anchors. Thus, in particular, they agree on the \textit{first} committed anchor (in this case $A2$, since $A1$ is skipped). Let the round of this anchor be $r'$.
\sysname opts to \textit{skip} all tentative anchor candidates in rounds $<r'$.

% \sasha{Is "the \textit{first} committed anchor" part clear? it refers to the Bullshark instance anchors. I had a similar concern while adjusting the description of run\_Bullshark above.}\fs{I tried to re-write to me more clear} 
%Summary: If anchor is faulty => then we will have to wait out, but then will skip ahead. If it is not faulty, then this whole scenario is unlikely.

Instead of pre-assigning anchors to fixed slots, \sysname considers all but the first anchor in each round \textit{virtual}, and dynamically materializes them as it sees fit. At any given time, there is only a single materialized consensus instance (with an anchor every other round); upon resolving the current instance, \sysname re-evaluates where to sensibly place the next anchor. Only then it begins evaluating the associated consensus instance. Since all replicas resolve each instance identically, each locally re-interprets and materializes the same anchor schedule.
Consider a slightly modification of the example in Figure \ref{fig:multipleanchors}, in which the yellow anchor does \textbf{not} extend the red anchor.  The blue anchor commits directly, and \sysname dynamically materialized the red node as the next anchor to evaluate. The red anchor, however, will be resolved as skipped upon committing the yellow anchor. \sysname thus skips evaluating the virtual green and pink anchors, and instead tries to materialize as anchor the first node \textit{after} the yellow anchor. 

\subsection{Operating multiple DAGs in parallel} %\fs{check what I wrote.} 
%\sasha{Looks good.}
Finally, in an effort to minimize Queuing latency, \sysname chooses to operate $k$ concurrent DAG instances, and interleave their outputs to construct a single, totally ordered log. Intuitively, staggering multiple DAGs, in a pipline-like manner, allows transactions to be proposed quicker. Instead of waiting (up to 3md) for the next DAG round to begin, a transaction simply joins the next round of the \textit{next DAG}. This reduces the expected Queuing latency from $1.5$md to $1.5/k$ md.
In practice, we deploy \sysname with three DAGs, each offset by one message delay. Since each DAG round requires 3md to certify a node proposal, this ensures that a some DAG's proposal is available every 1md. We find that 3 DAGs offers a good sweetspot between theoretical Queuing latency and batch sizes; additional DAGs offer diminishing returns as proposals carry increasingly smaller batches.

Algorithm~\ref{alg:shoalpp} summarizes the algorithm. Each DAG instance runs as a black-box, isolated from the others. In order to create a consistent total order, we must interleave their outputs. Each time a DAG commits an anchor, it outputs a new log segment.% to a queue.
%\sasha{nit: we do not really have a queue in the pseudocode.}\fs{I know: I wanted to somehow visualize that we are not "blocking" DAG progress to interleave. I.e. if one DAG is faster, it can keep adding things to the queue (rather than stopping and waiting), but they will only be picked off later.}
%\sasha{I think this is good.}
\sysname rotates across DAGs in round-robin fashion, and appends, in each round, exactly one available segment to the log. If, for whatever reason, one DAG commits more frequently, then it's excess available segments must wait to be ordered until the other DAGs to make progress. Importantly, the individual DAGs themselves never block, and progress unimpeded.
We note that segments also may not (and need not) be of equal size; each segment corresponds to a newly committed anchor, each of which may order varying sizes of causal histories.
%\fs{e.g. one is a direct commit, the other is a huge indirect commit, etc.}
%\sasha{I think this is clear.}
\input{code/shoalpp}

%  \begin{figure}[h]
%     \centering
%     \includegraphics[width=0.4\textwidth]{figures/threedags.pdf}
%     \caption{Three dags. Dashed squares are nodes. 
%     Solid squares are certified nodes. \sysname synchronize the DAGs in a one message delay offset.} 
%     \label{fig:threedags}
% \end{figure}

Operating additional DAGs incurs additional messages and associated processing overhead. However, we find that in practice, this does not negatively affect performance. In fact, throughput increases because smaller batches are proposed more frequently which helps amortize bandwidth and processing, resulting in better resource utilization. 
In theory, messages from different DAGs can be overlayed in order to share signatures. %(Fig. \ref{fig:threedags}).\fs{imo fig can be removed?} 
For instance, a node-proposal of DAG $i \in \{0,1,2\}$ can easily be combined with the certified-node-broadcast message of DAG $(i+1)\%3$.\footnote{In theory, votes can also be overlayed, though typically it is favorable to send them as fast as possible, rather than synchronizing.} This too, we find not to be worth it in practice as \one signatures are not the bottleneck, and \two message stages have assymetric latencies, and artificially aligning them increases Queuing latency.

\subsection{Discussion}
The combination of the above techniques allows \sysname to, in the common, fault-free case, reduce average e2e latency to only 4.5md: 4md to commit anchors, and 0.5md queuing latency. In the presence of faults or bad network behaviors, additional anchoring and commit latency will be incurred to resolve anchors via the Indirect Commit/Skip Rule.

We note that \sysname uses a \textit{linear} star-based communication pattern to ceritfy nodes, while protocols such as PBFT~\cite{castro2002practical} leverages \textit{quadratic} all-to-all communication. If desired, \sysname can adopt all-to-all communication as well, reducing latency by 1md (to 3.5md total) at the cost of increased message complexity. Finally, we remark that although \sysname incurs a small additional queuing latency, every replica acts as proposer, and clients need only contact their \textit{local} replica (in absence of faults). Single leader designs, in contrast, require clients to contact possibly remote leaders, introducing additional "queuing" latency.

%% file: code/dagapi.tex
\begin{algorithm}[h]
\caption{DAG interface}
\begin{algorithmic}[1]
\footnotesize

\Statex DAG API:        

\StateX \textsc{process\_node}(v: Node)
\StateX \textsc{process\_certified\_node}(v: Node)
\StateX \textsc{causal\_history}(v: Node) $\rightarrow$ Vec<Node>
\StateX \textsc{run\_bullshark}(anchor: Node) $\rightarrow$ Node
\StateX \textsc{get\_anchors}(round: int) -> Vec<Node>

\Statex

\Statex struct $\textit{Node } \mathsf{n}$: %\Comment{The struct of a node in the DAG}
        \StateX $ \mathsf{n}.\textit{round}$  \Comment{the associated round in the DAG}
        \StateX $ \mathsf{n}.\textit{source}$  \Comment{the replica proposing $n$}
        \StateX $ \mathsf{n}.\textit{parents}$  \Comment{$n-f$ nodes from round
        $ \mathsf{v}.\textit{round}-1$}

\end{algorithmic}
\label{alg:dagapi}
\end{algorithm}

%% file: code/mako.tex
\begin{algorithm}[h]
\caption{Core \sysname utilities}
\begin{algorithmic}[1]
\footnotesize

\Statex \textbf{Local variables:}
      
        \StateX dag: DAG 
        \Comment{Initially empty}
        \StateX weak\_votes: Vec< Vec<int> >
        \Comment{Initially empty}
        \StateX round: int
        \Comment{Initially 0}
        \StateX anchors: Vec<Node>
        \Comment{Initially empty}
\Statex

\Procedure{\textsc{next\_ordered\_nodes}()}{}

\While{true}

\If{anchors.\textsc{is\_empty}()}

    \State round++
    \State anchors $\leftarrow$ dag.\textsc{get\_anchors}(round)

\EndIf

\State anchor $\leftarrow$ anchors.\textsc{pop}()
\label{line:anchor}

\State anchor\_to\_order  $\leftarrow$ \textbf{select}
\label{line:anchortoorder}
\State \hspace{2cm} \textsc{fast\_commit}(anchor)
\State \hspace{2cm} dag.\textsc{run\_bullshark}(anchor)

\If{anchor\_to\_order $\neq$ anchor}
    %\State anchors.clear() \fs{ADDED anchors.clear to account for the skipping logic. TODO: round number should be passed to the function, and not be incremented here. Because maybe we skipped to a large one?}
    %\sasha{Updated the pseudocode with the skipping logic. We can save space by inlining the skip\_to function but I thought it is more readable like this.}
    \State \textsc{skip\_to}(anchor\_to\_order)

\EndIf
    \State \textbf{return} dag.\textsc{causal\_history}(anchor\_to\_order)
    \label{line:return}

%\Else    
%    \State return dag.\textsc{causal\_history}(anchor)
%\EndIf
\EndWhile

\EndProcedure

\Statex

\Procedure{\textsc{skip\_to}}{anchor: Node}

    \State round $\leftarrow$ anchor.round
    \State anchors $\leftarrow$ dag.\textsc{get\_anchors}(round)
    \State anchors.\textsc{remove}(anchor)

\EndProcedure

\Statex

\Procedure{\textsc{fast\_commit}}{n: Node}
    \While{true}
        \If{weak\_votes[n.round][n.source] $\geq 2f+1$}
            \State \textbf{return} n
        \EndIf

    \EndWhile

\EndProcedure

\Statex

\Receiving{node n}
        \State $r \leftarrow n.round$
        \State $s \leftarrow n.source$
        \If{\text{n is the first node revived from s in round r}} 
            \For{\textbf{each} $parent \in n.parents$}
            \State weak\_votes[r][parent]++
        \EndFor
        \State dag.\textsc{process\_node}(n)
        \EndIf

\EndReceiving

\Statex

\Receiving{certified node cn}       
        \State dag.\textsc{process\_certified\_node}(cn)
        
\EndReceiving
    
\end{algorithmic}
\label{alg:mako}
\end{algorithm}

%% file: code/shoalpp.tex
\begin{algorithm}[h]
\caption{\sysname}
\begin{algorithmic}[1]
\footnotesize

\Statex \textbf{Local variables:}
      
        \StateX $d_1, d_2, d_3$: DAG instances 
        \Comment{Initially empty}

        \Loop
            \For{i=1 to 3}

            \State \textbf{output} $d_i$.\textsc{next\_ordered\_nodes}()  
            
            \EndFor

        \EndLoop

\end{algorithmic}
\label{alg:shoalpp}
\end{algorithm}

%\balaji{What is Mako above?}
%\sasha{good catch}

%% file: sections/4-proofs.tex
\section{Correctness}\label{sec:correctnes}

We now prove the Safety and Liveness of \sysname.
We separately prove each augmentation \sysname introduces through a reduction to the Safety and Liveness proofs of Bullshark~\cite{bullshark} and Shoal~\cite{spiegelman2023shoal}. To formalize the arguments, our proof relies on the following three properties:

\begin{property}
\vspace{-3pt}
\label{prop:history}
    For any node $n$, the call \textsc{causal\_history}(n) returns the same vector of nodes to all replicas.
\end{property}

\begin{property}
\vspace{-3pt}
\label{prop:bullshark}
All replicas commit the same anchors, and in the same order.
%If some replica directly commits an anchor, then all other correct replicas will commit the anchor (whether directly or indirectly) as well, and all replicas commit all anchors in the same order. 
\end{property}

\begin{property}
\vspace{-3pt}
\label{prop:leaderreputation}
    For any given round $r$ the leader reputation mechanism ensures that \textsc{get\_anchors}(r) return the same vector of nodes at all replicas. 
\end{property}

Property~\ref{prop:history} holds true for all Narwhal-based DAGs, and follows directly from the fact that all nodes in the DAG are certified; Properties~\ref{prop:bullshark} and \ref{prop:leaderreputation} follow, respectively, from the safety proofs of Bullshark and Shoal.
%Since a Byzantine replica cannot successfully certify two nodes for any given round, all replicas converge on the same view of the DAG. 

\paragraph{Fast Direct Commit Rule} Since \sysname can commit using both the Fast Direct Commit Rule and the existing (Bullshark) Direct Commit Rule, Liveness follows directly from Bullshark.   
We prove that the Fast Direct Commit rule is safe: 

\begin{lemma}
\label{lem:das}
    For any anchor $a$, the procedures \textsc{fast\_commit}($a$) and \textsc{run\_bullshark}($a$) never return contradictory values.  
\end{lemma}

\begin{proof}\phantom{\qedhere}
By the code, \textsc{fast\_commit($a)$} can only return $a$.
Thus we need to prove that if \textsc{fast\_commit($a)$} triggers, then \textsc{run\_bullshark($a$)} also returns $a$. \textsc{fast\_commit($a)$} implies that the replica received $2f+1$ (uncertified) node proposals from different replicas that reference $a$.  Since there are at most $f$ faulty replicas, at least $f+1$ of these proposals were broadcasted by correct replicas and will eventually be certified. Since $f+1$ such certified nodes exist, some replicas may directly commit the anchor $a$ via the Bullshark logic. 
%the anchor $a$ is guaranteed to be directly/indirectly committed by the Bullshark logic. 
It follows from P.\ref{prop:bullshark} that in this case, the \textsc{run\_bullshark($a$)} call returns $a$ to all replicas. \qed
% The procedure \textsc{fast\_commit} never returns false, so we only need to prove that if it returns true, \\then \textsc{run\_bullshark}($a$) will not return false.
% Consider that \textsc{fast\_commit}($a$) returns true for some node $a$. This means that the replica received $2f+1$ (uncertified) nodes from different replicas. Since there are at most $f$ byzantine replicas, at least $f+1$ of these nodes were broadcasted by honest replicas and will eventually be certified.
% Therefore, some honest replica might receive $f+1$ certified nodes that link to $a$ and directly commit it. 
% Hence, the lemma follows from Property~\ref{prop:bullshark}.  
\end{proof}

\paragraph{Multiple Anchors per Round} Liveness follows directly from Bullshark. We show that \sysname's dynamic anchor mechanism retains Safety.

\begin{lemma}
\label{lem:siblinghelp}
   We assume all correct replicas are initialized with the same state (\emph{round} and \emph{anchors}). Then, all correct replicas' invocations of \textsc{next\_ordered\_nodes} (Alg~\ref{alg:mako}) return the same result. %, and (2) the local \emph{round} and \emph{anchors} variables across all replicas after the calls return are the same.

\end{lemma}

\begin{proof}\phantom{\qedhere}
%\fs{I'm a bit confused at this paragraph. It does not capture the skipping.}
%\sasha{because I did not have skipping as part of the pseudocode before. It was a practical consideration previously. }
By the lemma assumption and Property~\ref{prop:leaderreputation}, anchors.\textsc{pop}() returns the same anchor in the first iteration of the while loop to all replicas (Line~\ref{line:anchor}). 
By Lemma~\ref{lem:das}, all correct replicas agree on which anchor to order (Line~\ref{line:anchortoorder}). It follows, from the determinism of \textsc{next\_ordered\_nodes} and P.\ref{prop:history} that all replicas return the same vector of nodes to be ordered in their first \textsc{next\_ordered\_nodes} call (line~\ref{line:return}).
In addition, by the determinism of \textsc{next\_ordered\_nodes}, all replicas end up with the same state (\emph{round} and \emph{anchors}) after the calls return.
Therefore, by applying the argument above inductively, it is easy to show that all replicas also return the same vector of nodes in all future calls. \qed

% If the anchor is committed, then it follows from Property~\ref{prop:history} and the determinism of the code that the next call to  \textsc{next\_ordered\_nodes} is consistent at all correct replicas too.
% Otherwise, by Property~\ref{prop:leaderreputation} again, all replicas get the same anchor in the next iteration of the while loop as well. The lemma follows by \changebars{induction.}{repeating the above logic until the calls return.}
    
\end{proof}

% \begin{lemma}
% \label{lem:mako}
%      The \sysname algorithm presented in~\ref{alg:mako} satisfies Safety.
% \end{lemma}

% \begin{proof}

% We need to prove that for any $k \geq 1$, 
% the $k^{th}$ call of \textsc{next\_ordered\_nodes} by honest replicas returns the same vector of nodes.

% The lemma follows directly by inductively applying Lemma~\ref{lem:siblinghelp}.
    
% \end{proof}

\begin{theorem}
    The \sysname protocol Algorithm~\ref{alg:shoalpp} satisfies Safety and Liveness.
\end{theorem}
The Liveness and Safety proof of each DAG is given above. Since DAG outputs are interleaved deterministically, Safety and Liveness of the composed system follows immediately.

%% file: sections/5-practical-new.tex
\section{Practical Considerations} 
\label{sec:practical}
We briefly discuss some implementation considerations.

\paragraph{Inline data streaming} 
Narwahl~\cite{narwhaltusk} proposes, in addition to the DAG construction, to further decouple, and horizontally scale data dissemination using an additional "worker"-layer. Workers disseminate batches optimistically a-priori, and DAG proposals contain only hash references. This construction is extremely throughput-friendly in the best case, as data is streamed continuously, and consensus carries only metadata. In practice, however, much like with uncertified DAG-BFT protocols~\cite{keidar2022cordial, mysticeti}, other replicas might not have the data behind the hashes (or the references in the case of uncertified DAGs) when the proposal is received and thus additional 2md latency is required for fetching.  
%To ensure robustness however, proposals may only include batches that have been reliably broadcast. This introduces an additional 2md end-to-end latency. 
\sysname opts to forgo the worker layer in order to avoid this latency, and chooses to disseminate transaction batches inline with DAG proposals. Operating staggered DAGs allows \sysname to nonetheless achieve throughput competitive to that of Narwahl (with a single worker), as smaller batches are sent more frequently, thus amortizing bandwidth close to a streaming design.\footnote{Scaling throughput by employing more workers is typically not worth it in existing BFT deployments, as single-worker ordering throughput in existing protocols already far exceeds execution bandwidth.}
Finally, we note that inline data streaming simplifies data fetching as it avoids an additional indirection needed to dereference digests.

\paragraph{Efficient fetching}
DAG edges in \sysname contain only certified nodes, allowing \sysname to efficiently fetch missing nodes and resolve inconsistencies across replicas. Replicas whose DAG views are out of sync, may validate and vote on proposals \textit{without} locally observing the proposed node's causal history. This streamlines the certification process, and ensures that the DAG operates at predictable rates. Any missing causal histories can be inferred asynchronously, off the critical path. We remark that this is not possible in uncertified DAGs; to ensure liveness, replicas must fetch all missing data \textit{before} processing a new node. This results in unsteady processing times, and may significantly increase e2e latency.
%\fs{If there is space can add the example figure, otherwise I'd say omit.}
Certified edges, additionally allow us to balance data fetching load: since at least $f+1$ correct replicas must have observed any certify node, replicas missing the node may request it from different replicas to balance load.
\iffalse
An illustration of the uncertified DAG's latency issues appears in Figure~\ref{fig:uncertifieddag}.

\begin{figure}[h]
    \centering
    \includegraphics[width=0.4\textwidth]{figures/uncertifieddag.pdf}
    \caption{Solid squares represent nodes that are locally available to validator 1, out of which the black squares represent its local DAG. In an uncertified DAG, validator 1 cannot advance to round 2 and broadcast node $n''$ until node $n'$ is added to the DAG. But $n'$ cannot be added to the DAG until node $n$ is received. This is because no certificate proves its existence. In this example, validator 3 behaves honestly, and yet its node cannot be added, and the DAG construction stalls. 
    Moreover, to fetch node $n$, all validators have to contact validators 3 or 4 (who become a bottleneck) as there is no information about who else has the node.} 
    \label{fig:uncertifieddag}
\end{figure}
\fi

\paragraph{Distance-based priority broadcast}
We observe that in systems with large $n$, message broadcasting consumes considerable CPU times. When sending messages in a fixed order (e.g. first to $R1$, then to $R2$, ...) can unintentionally cause replicas to consistently receive messages with delay, making their nodes subject to falling behind in the DAG.
To circumvent this issue, \sysname, periodically records point latencies between replicas, and adjust broadcast send-orders to prioritize farther away replicas. This results in a more balanced message distribution, increasing the homogeneity of node certification pace.

%\fs{optional: Could talk about leader reputation}%(4) Leader reputation: Give a summary?)
%\sasha{Nothing really to say on top of Shoal}

%% file: sections/7-evaluation.tex
\section{Evaluation} 
\label{sec:evaluation}
%\fs{TODO: Change figure names to Bolt. Also in breakdown, get rid of the +1, +2 thing. Say: Shoal, Bolt-FastCR, Bolt-FastCR + MultiAnchor, Bolt.}
%\balaji{We only order nodes that have transactions to reduce VM overhead.}
%\sasha{Should we talk about how we sync the DAGs?}

%\sasha{Explain why we compare to Jolteon and not PBFT. Not chained and not clear how to do pipeline?}
%\fs{I like framing it as questions.}
Our evaluation seeks to to answer three questions:
\begin{enumerate}
    \item \textbf{Latency reduction:} How much does \sysname reduce Latency compared to state-of-the-art alternatives?
     \item \textbf{Latency breakdown:} How much does each of the augmentations introduced in \sysname contribute to the latency reduction?
    \item \textbf{Robustness:} How well does \sysname tolerate faults?
\end{enumerate}
%  With our evaluation, we aim to demonstrate
% \begin{enumerate}
%     \item High throughput that is comparable to state-of-the-art DAG BFT protocols.
%     \item Lower latency compared to alternative DAG BFT protocols and comparable to state-of-the-art leader-based ones.   
%     \item Breakdown showing the latency improvement achieved by each of \sysname's contributions. 
% \end{enumerate}

We implemented \sysname into an open-source code base of Aptos, one of the major blockchain projects. Our prototype is written in Rust and utilizes the Tokio~\cite{tokio} 
asynchronous runtime. It uses BLS~\cite{bls2001} implemented over BLS12-381 curves for signatures, 
RocksDB~\cite{rocksdb} for persistent storage of consensus data, and Noise~\cite{noise}
for authentication.
%\fs{what is authentication used for?}
We disable the execution and ledger storage components of the original blockchain in order to isolate consensus performance. %However, we persisted all consensus data residing in network-attached SSDs.
Each data point represents the 50 percentile (median) with error bars representing the 25, and 75th percentiles.

 \begin{figure*}[t]
    \centering
    \includegraphics[width=0.95\textwidth]{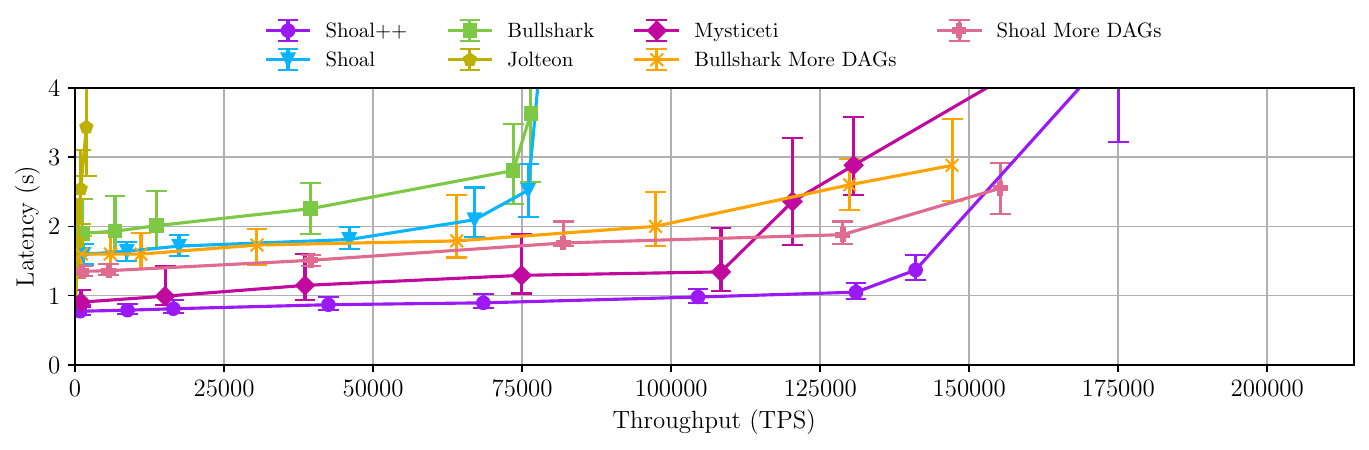}
    \vspace{-10pt}
    \caption{No failures with 100 replicas.}
    \label{fig:no-failures-macro}
\end{figure*}

\paragraph{Baselines} We compare \sysname against three popular high throughput DAG-BFT consensus protocols, Bullshark~\cite{bullsharksync}, Shoal~\cite{spiegelman2023shoal}, and Mysticeti~\cite{mysticeti}, as well as one traditional low-latency BFT protocol, Jolteon~\cite{jolteon}.
Bullshark and Shoal represent state-of-the-art \textit{certified} DAG-based protocols, while Mysticeti~\cite{mysticeti} is a concurrent work to \sysname that proposes an \textit{uncertified} design to reduce latency.
% Bullshark, which is currently deployed on the Sui Blockchain, and Shoal are the only two existing partially synchronous DAG BFT protocols providing experimental results. 
% Mysticeti is a concurrent DAG BFT work that also aims to improve latency.
Jolteon~\cite{jolteon} is a state-of-the-art leader-based BFT protocol, which improves Hotstuff~\cite{hotstuff} latency by 50\%.
A variant of Jolteon is currently deployed on Aptos~\cite{aptos}.  Sui~\cite{sui} recently replaced Bullshark with a Mysticeti deployment. 
%Bullshark and Jolteon, respectively, are currently deployed on the Sui~\cite{sui} and Aptos~\cite{aptos} Blockchains.
For an apples-to-apples comparison, we re-implemented Bullshark, Shoal, and Jolteon according to the papers' description in the same codebase as \sysname. For Mysticeti we run the publicly available source code~\cite{mysticeti-code} referred to in the latest version of the paper; its prototype too is written in Rust, using Tokio, but forgoes writing consensus data to persistent storage, making it less production-realistic. 
To minimize latency, we opt to disable the Narwhal~\cite{narwhaltusk} worker layer for both Bullshark and Shoal. However, for a fair throughput comparison with \sysname, we additionally augment both systems with \sysname's proposed parallel-DAG strategy ("Bullshark/Shoal More DAGs", Fig.~\ref{fig:no-failures-macro}).

%\fs{check this.}
%Note we only compre to Mysticeti's consensus not the fast path.

%For the latency breakdown (bullet three), we micro-benchmarked \sysname with ... 
 
\paragraph{Experimental setup}
 We use the Google Cloud Platform~\cite{gcp} to mimic the deployment of a globally decentralized network. Our testbed consists of 100 replicas spread evenly across 10 regions around the world: we choose two regions in the US (us-west1 and us-east1), two in Europe (europe-west4 and europe-southwest1), three in Asia (asia-northeast3, asia-southeast1, and asia-south1), and, respectively, one each in South America (southamerica-east1), South Africa (africa-south1) and Australia (australia-southeast1).
 The round-trip times between these regions range between 25ms and 317ms.
 %\fs{I feel like you could replace this with a region to region latency matrix.}
 We use \texttt{n2d-standard-64} virtual machines, containing 64 vCPUs and 256 GB of memory~\cite{machines}.
 %\footnote{\url{https://cloud.google.com/compute/docs/general-purpose-machines\#n2d_machines}} 
 We provisioned a 2TB network attached disk to each machine to guarantee enough IOPS for persistence.
 This spec is similar to those used by production blockchains and qualifies as commodity grade.%\fs{what does this mean?}. 
% \sasha{Generally available and not too expensive}
 %\fs{check this next bit:}
 
Clients connect to a single (local) replica and issue a continuous stream of dummy transactions (310 random bytes). As is standard in blockchain deployments, we measure consensus latency as the time between when a transaction first arrives at a replica, and when the transaction is ordered and ready to be executed.
%\footnote{Most blockchain deployments do not reply to clients; instead, users may poll state from local replicas.}
 %the DAG node containing the transactions is ordered. %Note that the time from when a transaction is added to the mempool until when it is proposed in a DAg node is denoted as queuing latency.
 We use a batch size of 500 transactions across all systems. In all DAG-based systems, each node proposal contains one batch; Jolteon proposals may contain up to 100 batches.
 %\fs{what about mysticeti? is this correct?}
 %We allow DAG node proposals in \sysname to batch up to 500 transactions.\fs{what about hte other systems?}
 % \sasha{We do the same for all implementations}
%\sasha{Not sure about the bathcing part in Mysteciti}
We configure Jolteon to use a 1.5s timeout for view changes (the standard in existing production systems~\cite{aptos}), and use a round timeout of 600ms for Bullshark, Shoal and \sysname; Mysticeti by default uses a 1s round timeout. We remark that timeouts in Bullshark, Shoal and \sysname serve different functions. Bullshark requires them for liveness, while Shoal manages to sidestep using timeouts for liveness all but edge cases; \sysname follows Shoal, but re-introduces a timeout only for improved performance. We note that the timeout is started at the \textit{beginning} of a round, and thus typically introduces only a small additional wait past observing the first $2f+1$ certified nodes.

\subsection{Failure-free case}
\label{sec:nofailures}
%\sasha{How many replicas? add to all sections}
%\sasha{TODO: change tx/s to TPS. talk about p50, p25, p75 in the graphs.}

Figure~\ref{fig:no-failures-macro} illustrates throughput and latency under failure-free conditions. \sysname is the only system to sustain sub-second latency for throughput of 100k transactions per second (tps).% while simultaneously scaling throughput to over 150k tps. 
%\fs{TODO: in the following say exact numbers, and factors of improvement. 775ms under low load}

Jolteon offers low latency (900ms) at low load, but is unable to scale throughput beyond 2100 tps as it is bottlenecked by the leader's network bandwidth. Bullshark and Shoal support significantly higher throughput (up to 75k tps) but incur high latency: 1.9s and 1.45s, respectively, at low load, and 2.4s and 1.7s, respectively at 50k tps. This follows directly from their high average number of required message delays (\S\ref{sec:dissect}). Throughput scalability in both systems is bottlenecked by network utilization. 
%and data processing speeds, in particular message delivery, and persistence to storage.\fs{Confirm that what I say is the actually bottleneck.}
%\balaji{However, for completeness, we also show that Shoal and Bullshark can both benefit from the "More DAGs" optimization and match \sysname's throughput without the need for the Narwhal streaming optimization.}

\sysname significantly reduces latency (775ms at low load), while simultaneously improving throughput scalability (up to ca. 140k tps).
The parallel DAG construction allows \sysname to better utilize network and storage resources: rather than broadcasting and processing a large batch once per round, it broadcasts and processes three smaller batches at the same time. 
Applying the same technique to Bullshark and Shoal ("More DAGs") allows both systems to match the throughput of \sysname (while simultaneously improving queuing latency), without incurring the latency penalty of a Narwahl-style worker layer.

%\fs{Bullshark and Shoal++ are missing the ankle points. From this graph it looks like Shoal++ can get even more tput?}
%\fs{why is Shoal tput so much worse than bullshark? that doesn't make sense to me.}
%\sasha{Grphas are not done yet. Should be fine.}\fs{ok. Remove comments once you've resolved it.}

%\fs{say: under best conditions mysticeti has good latency. But why is its latency not as good as Bolt? Probably because it ends up having to data sync on the critical path a decent amount even without faults, degrading latency. It's tput scales further because it doesnt store to SSD.}
%\sasha{Mysticeti does not have a higher TP. Grpahs are not final.}
Mysticeti matches the throughput of \sysname, but fails to match its latency, even at moderate loads. At high loads (upwards of 100k tps) Mysticeti notably incurs a signficant latency penalty as replicas may be temporarily out of sync, and must fetch missing data on the critical path of consensus.% Mysticeti offers the highest throughput of all systems because it does not write consensus data to persistent storage.
%\sasha{Mysticeti does not have a higher TP. Grpahs are not final.}\fs{ok, replace as you see fit later.}

\begin{figure}[t]
    \centering
    \includegraphics[width=0.95\columnwidth]{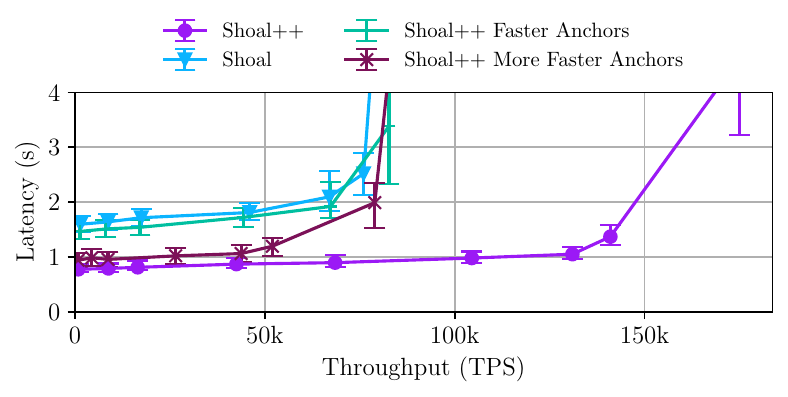}
    \vspace{-10pt}
    \caption{\sysname Breakdown. No failures with 100 geo-distributed replicas.}
    \label{fig:no-failures-micro}
\end{figure}

% and can support up to 150k tps, but failed to match the Jolteon latency
% \sysname performs best with the highest throughput and lowest latency. \fs{TODO: say what tput it can reach, and why: It can scale tput beyond Bullshark/Shoal because it better amortizes resources by having 3 DAGS; akin to streaming.}

\begin{figure}[t]
    \centering
    \includegraphics[width=0.95\columnwidth]{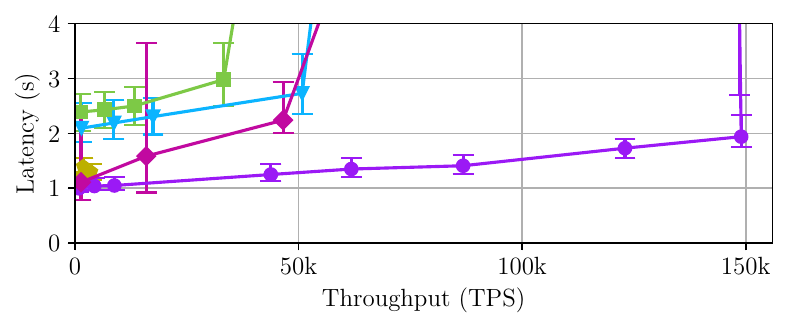}
    \vspace{-10pt}
    \caption{Latency/Throughput graph with 33 out of 100 crashed replicas.}
    \label{fig:failures-macro}
    \vspace{-10pt}
\end{figure}

\subsection{Latency improvement breakdown}
\label{sec:breakdown}

Figure~\ref{fig:no-failures-micro} isolates the respective latency improvement of each \sysname augmentation compared to the baseline, Shoal. \sysname \textit{Faster Anchors} shows Shoal augmented only with the Fast Direct Commit Rule while \sysname \textit{More Faster Anchors} additionally applies the multiple anchor augmentation. Finally, \sysname includes also the parallel DAG optimization.

The Fast Commit rule, in theory, improves anchor commit latency by up to 2md. In real, asymmetric network settings, however, it is not always 2md faster than the Direct Commit rule, and thus the practical latency improvement is smaller. The multiple anchor augmentation, in turn, improves latency significantly: \one it saves, on average, 3md by eliminating the anchoring latency for most nodes (those that become anchors), and \two ensures that the anchoring latency for non-anchor nodes is typically only one DAG round because (with overwhelming probability) nodes are referenced by at least one of the (many) anchors in the next round.   
%\fs{check!} 
The parallel DAG framework improves latency further, by reducing queuing latency, but more importantly improves throughput scalability by approaching a data "streaming" effect; crucially, and unlike Narwhal-based worker designs, it does so \textit{without} incurring additional latency.

% Bolt-FastCR, Bolt-FastCR + MultiAnchor, Bolt
% [Shoal+1] is reducing the anchor latency time.
% [Shoal+1] combines anchor latency and anchoring latency improvement.
% As expected, eliminating the anchoring latency gave the most latency improvements as it saves 3md.
% %The anchor latency improvement saves 2md delays in theory, but the graph suggests the practical improvement is smaller. 
% %One reason might be that
% The queuing latency improvement not only improved the latency but also significantly improved throughput. As described above, to achieve better latency we forgo the data streaming mechanism of Narwhal and instead inlined the data inside the nodes. The three parallel DAG constructions provide the "streaming" effect by casting data more frequently.  

%\sasha{TODO: do not forget the change ledger to Bolt}

\subsection{Failure case}
\label{sec:failures}

To quantify the robustness of \sysname, relative to the baseline systems, we evaluate performance under two types of disruptions. 

First, we measure the performance of all systems while simulating crash failures of 33 (out of 100) replicas (Fig.~\ref{fig:failures-macro}). The performance of Jolteon remains largely unaffected, as its leader reputation mechanism quickly detects failures and subsequently elects only alive replicas. Shoal and \sysname, likewise, swiftly adjust and elect only live anchors; latency, however, increases (up to ca. 2x at high load) as quorums must span more regions, increasing tail latency.
Bullshark and Mysticeti, in contrast, suffer drastically increased latency since they do not employ a leader reputation mechanism to elect suitable anchor candidates, and thus must "wait" to commit until a non-faulty replica is elected as anchor. 
%\fs{Increase y axis. I want to see }
%\sasha{camera ready}
%Due to the leader reputation   
%DAG-based BFT protocols in contrast continue to disseminate batches and grow the DAG. If an anchor is faulty, latency naturally increases, but throughput remains high as nodes disseminated during the anchor interruption are committed as causal histories. 

Second, we measure the impact of sporadic message drops for \sysname and Mysticeti (Fig.~\ref{fig:packet-drop}). We inject network layer message drops for 1\% of egress traffic in 5 nodes (out of 100 total), beginning at 60 seconds (red line). Under normal network conditions and a low-moderate load (18k tps) both \sysname and Mysticeti offer 700ms median consensus latency. Upon failure injection the latency of Mysticeti rises sharply (by a factor of 10x) as replicas scramble to perform critical-path synchronization on missing data (resulting, at times, in timeouts). As latency spikes, throughput initially drops as well, but recovers once the backlog of disseminated DAG proposals commit. \sysname, in contrast, remains largely unaffected by message drops (latency rises to at most 1.3x). Because all nodes are \textit{certified}, the DAG construction process in \sysname proceeds smoothly; any required synchronization is asynchronous and off the critical path.

% \begin{figure}[t]
% \centering
% \begin{subfigure}{\columnwidth}
%     \includegraphics[width=\columnwidth]{plots/shoalpp-drops-timeline.pdf}
%     \vspace{-20pt}
%     \caption{\sysname (Certified DAG)}
%     \label{fig:shoalpp-packet-drop}
% \end{subfigure}
% \begin{subfigure}{\columnwidth}
%     \includegraphics[width=\columnwidth]{plots/mysticeti-drops-timeline.pdf}
%     \vspace{-18pt}
%     \caption{Mysticeti (Uncertified DAG)}
%     \label{fig:mysticeti-packet-drop}
% \end{subfigure}
% \caption{Impact of message drops in Certified vs Uncertified DAGs}
% \label{fig:packet-drop}
% \end{figure}

% \balaji{
% Figure~\ref{fig:packet-drop} shows the impact of message drops in Certified DAG construction as in \sysname and Uncertified DAG construction as in Mysticeti, respectively. For this experiment, we injected network layer outbound message drops in about 5 nodes (out of 100 total) about 1\% of the time. The time of failure injection is marked with a vertical red dotted line. Under normal network conditions, both \sysname and Mysticeti provide optimal latency of about 700ms at a low-moderate 18k TPS.  As soon as the failure is injected, the benefits of the certified DAG is apparent as the latency of Mysticeti shoots up at least 10$\times$ while \sysname's latency only increases about 30\% at most. 
% }

\begin{figure}[t]
\centering
\includegraphics[width=\columnwidth]{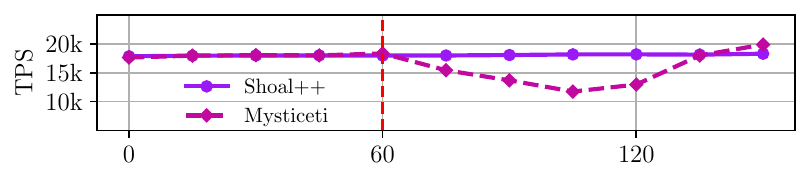}
\includegraphics[width=\columnwidth]{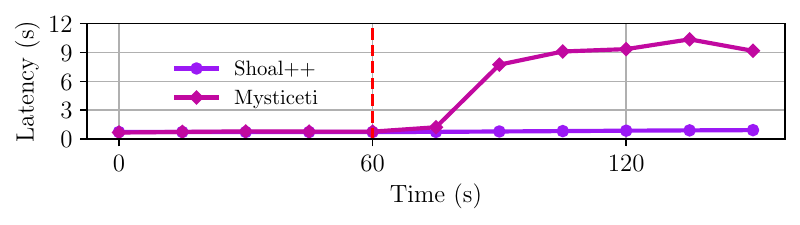}
\vspace{-22pt}
\caption{Impact of message drops in Certified vs Uncertified DAGs.}
\label{fig:packet-drop}
\vspace{-12pt}
\end{figure}

% \begin{figure}[t]
% \centering
% \begin{subfigure}{\columnwidth}
%     \includegraphics[width=\columnwidth]{plots/drops-timeline-tps.pdf}
%     \vspace{-20pt}
%     \caption{TPS}
%     \label{fig:shoalpp-packet-drop}
% \end{subfigure}
% \begin{subfigure}{\columnwidth}
%     \includegraphics[width=\columnwidth]{plots/drops-timeline-lat-log.pdf}
%     \vspace{-18pt}
%     \caption{Latency (s)}
%     \label{fig:mysticeti-packet-drop}
% \end{subfigure}
% \caption{Impact of message drops in Certified vs Uncertified DAGs}
% \label{fig:packet-drop}
% \end{figure}

%% file: sections/8-related.tex
\section{Related work} \label{sec:related}

%\fs{proposed re-write; old is commented out below}
\paragraph{Latency optimal BFT} Spurred by the seminal PBFT~\cite{castro2002practical} protocol, a long line of systems~\cite{zyzzyva, sbft, martin2006fast, suri2021basil, jolteon} optimize for consensus latency. PBFT-based protocols designate a single leader to act as sequencer, and replace it only sporadically (e.g upon failure, or eagerly for fairness). This approach enjoys excellent latency. PBFT  requires, when led by a correct replica, only 3 message delays (md); the optimum achievable latency in presence of faults, and optimal resilience of $n=3f+1$. Several followup works suggest optimistic Fast Paths~\cite{zyzzyva, sbft, martin2006fast, suri2021basil} that improve latency to only 2md in absence of faults, or by weakening resilience to $n=5f+1$.
Unfortunately, single leader-based designs are fundamentally bottlenecked in throughput by the bandwidth and processing capacity of a single replica.

\paragraph{High throughput BFT} 
The key idea to scale beyond this bottleneck
was realizing that networking resources must be fully utilized.
Some works propose to decouple data dissemination from the consensus logic~\cite{narwhaltusk, dispersedledger, arun2022dqbft}, while others~\cite{MirBFT, gupta2021rcc, stathakopoulou2022state} propose multi-log frameworks: these systems partition the request space and operate multiple parallel black-box instances of consensus protocols (e.g. PBFT) -- each led by a different replica --, and carefully intertwine their outputs to construct a single, totally ordered log. This approach can achieve high throughput but requires complex coordination to deal with failures and re-configuration.

\paragraph{DAG BFT}
More recently, a new class of so-called DAG-BFT protocols have emerged as the popular choice for high throughput BFT. At a high level, deployed DAG-BFT protocols propose to explicitly separate data dissemination from consensus, and utilizing only a single common DAG data structure to implement consensus (\S\ref{sec:background}).
%\fs{fit in dispersed ledger? I am not familiar.}

HashGraph~\cite{hashgraph} was the first to propose a DAG-based consensus structure, and Danezis et al.~\cite{embedding} showed that any deterministic BFT protocol can, in theory, be embedded on a DAG. Aleph~\cite{aleph} and DAG-Rider~\cite{allyouneed} were the first to introduce additional rigidity to the DAG construction, leveraging a round-based structure to design asynchronous BFT protocols. Blockmania~\cite{blockmania} was the first attempt to implement a DAG-BFT system, but fell short of adoption as it inherited a complex view change procedure, notorious to traditional BFT protocols such as PBFT.

Narwhal and Tusk~\cite{narwhaltusk} catapulted DAG-BFT to popularity by demonstrating that \one (asynchronous) consensus can be embedded onto DAG-structures without the need for explicit view changes (Tusk), and \two throughput can be scaled nearly linearly through the addition of a horizontally scalable worker layer (Narwahl Mempool). Bullshark~\cite{bullshark} improved upon Tusk by adopting partial synchrony and introducing a common case fast path during synchronous intervals; in a future iteration~\cite{bullsharksync} Bullshark removed the asynchronous fallback, and established itself as the de facto partially synchronous DAG-BFT protocol.

Bullshark designates, at regular intervals (every two rounds), DAG-nodes to simulate a leader (anchor) akin to traditional BFT protocols, and commits them upon observing DAG patterns that ensure durable agreement of the leader.
We refer to \S\ref{sec:background} for an overview of the Narwahl core DAG, and the partially synchronous Bullshark consensus protocol. Shoal improves the latency of Bullshark by \one increasing anchor frequency to every round, \two utilizing a deterministic reputation scheme to select as anchors only the fastest, and best-connected replicas, and \three eliminating leader/anchor timeouts in all but extremely unlikely scenarios. Shoal, in a closing remark, briefly discusses the idea of multiple anchor per round; however, it neither provides an implementation, nor detailed consideration of practical details.
Unfortunately, best case (failure free) average latency of state of the art partially synchronous DAG-BFT protocols remains as high as 10.5md (Shoal). 

\iffalse
\sysname, in turn, reduces average latencies (in absence of faults) to an average of only 4.5md. We note that \sysname uses a \textit{linear} star-based communication pattern to ceritfy nodes, while PBFT leverages \textit{quadratic} all-to-all communication. If desired, \sysname can adopt all-to-all communication as well, reducing latency by 1md at the cost of increased message complexity. Finally, we remark that the additional 0.5md of latency \sysname incurs are due to Queuing Latencies in the DAG. In \sysname, however, every replica acts as proposer, and clients need only contact their \textit{local} replica (in absence of faults). Single leader designs, in contrast, require clients to contact possibly remote leaders, introducing additional "queuing" latency.
\sasha{Sometimes these discussions go to the discussion section}
\fi

\paragraph{Uncertified DAGs} A recent group of protocols propose the design of uncertified DAGs. Such constructions can, in theory, reduce latency, but are increasingly brittle as they are prone to data fetching on the critical path (\S\ref{sec:dissect}). Cordial Miners~\cite{keidar2022cordial} replaces the Reliable Broadcast (RB) in each round with best-effort broadcast (BEB), allowing anchors to (in the best case) commit within 3md; anchors are placed every other round, for an average latency of 4.5md. 
%\fs{3 for anchors, 5 for siblings, 4 for odd rounds. Check?} 
Similarly, BBCA-chain~\cite{bbca} replaces non-anchor nodes with BEB, and strengthens anchor nodes to single-shot pbft instances. Unfortunately, neither proposal has an implementation that empirically evaluates practical performance. 
%\fs{You were talking about BBCA-ledger => I replaced it with BBCA-chain (which is what you cite actually)}

\paragraph{Concurrent Latency reduction efforts}
%\paragraph{Concurrent efforts with \sysname to reduce latency.} 
In concurrent efforts to \sysname, several recent (non-peer reviewed) protocols suggest latency improvements, underlining the demand for high throughput systems with lower latency.
%\fs{added this sentence}

Sailfish~\cite{Sailfish} introduces a new anchor commit rule that resembles the Fast Direct Commit Rule proposed in ~\sysname, but unlike \sysname forgoes leveraging the Bullshark Direct Commit rule as potentially faster alternative. Unfortunately, Sailfish lacks an implementation that empirically substantiates their suggested latency improvement.

Mysticeti improves upon Bullshark's best case latency by transitioning to an uncertified DAG and implementing the Cordial Miners~\cite{keidar2022cordial} consensus protocol. Mysticeti, like \sysname, adopts Shoals high-level proposal to optimistically employ multiple anchors per round. 
Our evaluation demonstrates that Mysticeti almost matches \sysname's latency in the common case. However, as shown in \S\ref{sec:evaluation}, and independently corroborated by Autobahn~\cite{giridharan2024autobahn}, uncertified DAG constructions are prone to data fetching on the critical path, which negates the suggested latency benefits. A single Byzantine (or slow) replica, for instance, can impose at least 2 additional message delays \textit{per} round in Mysticeti by not fully (or timely) disseminating its node proposals. In \S\ref{sec:evaluation} we observed 10x latency degradation when 5 out of 100 nodes experienced 1\% message drop.  
%The certified DAG employed by \sysname significantly enhances its resilience to Byzantine behavior and bad networks.
%\fs{I put tentative sentences (subject to graphs) in changebars}

Autobahn~\cite{giridharan2024autobahn} introduces a promising DAG-\textit{free} consensus approach in which replicas efficiently and reliably disseminate data independently in parallel "data lanes", and consensus commits only snapshot cuts of the lane states.

%They report throughput that matches 
%existing certified DAG-BFT protocols while improving upon their latency.
%and simultaneously improving robustness to asynchrony. 
%We leave in-depth comparison to future work.
%\sasha{I think one of the Dumbo papers did something similar. Maybe we should cite them as well.}\fs{sure up to you; it's asynchronous though, and I don't think we should go there again}
%\sasha{If they are asynchronous never mind}

\paragraph{Leader reputation in BFT} 
To the best of our knowledge, leader reputation for BFT systems was first implemented in DiemBFT~\cite{diembft}, and formalized by Carousel~\cite{cohen2022aware}. Shoal~\cite{spiegelman2023shoal} and later Hammerhead~\cite{hammerhead} adopt leader reputation to DAG-BFT in an effort to exclude underperforming leaders from anchor candidacy. \sysname extends the approach to provide, in each round, a set of anchor candidates.

%% file: sections/9-conclusion.tex
%\balance
\section{Conclusion} \label{sec:conclusion}
This work introduces \sysname, a DAG-based BFT consensus protocol that offers high throughput, while significantly lowering e2e latency compared to current state of the art DAG-BFT protocols. \sysname takes as its starting point existing state of the art DAG-BFT protocols, and augments them in three key ways: \one it reduces, in the common case, the commit latency of leaders (anchors), \two dynamically reinterprets anchor schedules to increase anchor frequency, and \three operates multiple parallel DAG instances to increase proposal frequency. This, in sum, allows \sysname to reduce the fault-free average e2e latency from 10.5 message exchanges (Shoal~\cite{spiegelman2023shoal}) to only 4.5 message exchanges, within reach of traditional latency optimal (but throughput-inefficient) BFT protocols such as PBFT. 

We note that the parallel DAG instances technique might be of independent interest, as this idea can be applied to any BFT algorithm to reduce its queuing latency.

%\sasha{More DAGs good for more protocols}
%\sasha{Should we mention some numbers here?}